\documentclass[prb,aps,twocolumn,superscriptaddress,floatfix]{revtex4-2}

\usepackage{amsmath}
\usepackage{amssymb}
\usepackage{bm}
\usepackage{graphicx}
\usepackage{capt-of}
\usepackage{tikz}
\usetikzlibrary{calc}

\newcommand{\rcol}[1]{{\raggedright #1\par}}
\newtheorem{theorem}{Bogoliubov sum rule.}

\newtheorem{corollary}{Corollary}

\usepackage[colorlinks=true,linkcolor=blue,citecolor=blue,urlcolor=blue,bookmarks=false]{hyperref}

\graphicspath{{figures/}}

\begin{document}
	
	\title{Bogoliubov sum rule and the Knight-shift ellipsoid in spin-locked superconductors}
	
	\author{Yi Zhou}
	\email{yizhou@iphy.ac.cn}
	\affiliation{Institute of Physics, Chinese Academy of Sciences, Beijing 100190, China}
	
	\date{\today}
	
	\begin{abstract}
				We establish an exact Bogoliubov sum rule for any Hermitian single-particle operator $O$: at each momentum, its particle-hole and particle-particle matrix-element weights sum to the single-particle trace $\mathrm{Tr}_{s}(O^{2})$. The result follows solely from Hilbert--Schmidt-norm invariance under a canonical Bogoliubov transformation and contains neither excitation-energy denominators nor occupations; physical response identities therefore require additional assumptions. At zero field, a fully gapped helicity-diagonal state with both helicity sheets present and a spin-orbit splitting asymptotically larger than the gap obeys $\chi_{\mu\nu}(0)/\chi_N=\delta_{\mu\nu}-\Pi_{\mu\nu}+o(1)$, where $\chi_N$ is the normal-state Pauli susceptibility and $\Pi=\langle\hat{\mathbf n}_{\mathbf k}\hat{\mathbf n}_{\mathbf k}\rangle_{\rm FS}$ is the Fermi-surface average of the unit spin-locking texture. The eigenvalues of $\Pi$ form a simplex, while the normalized spin Knight-shift tensor defines an ellipsoid whose semi-axes are the residual principal responses. Full cubic invariance of both the superconducting state and locking texture fixes $\Pi=\mathbb I/3$ and hence $\chi(0)/\chi_N=2\mathbb I/3+o(1)$; cubic crystal symmetry alone does not. For zero-field $s$-wave pairing in the reference-Fermi-surface regime, we obtain the exact closed-form kernel $F_s(\lambda)=1-\operatorname{asinh}\lambda/[\lambda\sqrt{1+\lambda^2}]$, valid for arbitrary $\lambda=|\mathbf g|/\Delta$. In a finite Zeeman field, the zero-field helicity reduction generally fails, so the equilibrium magnetization and differential response require a self-consistent BdG calculation rather than a field-dependent locking-tensor substitution. Applied to the $^{75}$As data on K$_2$Cr$_3$As$_3$, the framework identifies a field-dependent axial suppression pattern at $8$--$16$~T. Both its placement near the $\hat c$-axis vertex and the resulting tension with a specified decoupled-pocket spin-orbit-texture baseline require an additional fixed-tensor $T\to0$ and zero-field extrapolation; no unique microscopic pairing state follows.
	\end{abstract}
	
	\maketitle
	
	\tableofcontents
	
	\section{Introduction}
	\label{sec:intro}
	
	The Knight shift is among the oldest and most direct probes of spin susceptibility in superconductors. Early measurements in mercury by Reif~\cite{Reif1957} and in tin by Androes and Knight~\cite{AndroesKnight1959} found a finite residual Pauli response at $T=0$, in contrast to the clean, spin-rotation-invariant BCS result~\cite{BCS1957,Yosida1958}. Ferrell~\cite{Ferrell1959}, Anderson~\cite{AndersonKnight1959}, Martin and Kadanoff~\cite{Kadanoff1959}, and Schrieffer~\cite{Schrieffer1959} showed how spin-orbit scattering can preserve a residual response; Appel~\cite{Appel1965}, Shiba~\cite{shiba1976effect}, Zhogolev and Glasser~\cite{zhogolev1972magnetic}, and Anderson's dirty-limit analysis~\cite{Andersontheoryofdirty1959} subsequently refined this picture. These developments established the residual Knight shift as a quantitative probe of the competition between pairing and spin-rotation breaking. Standard treatments of NMR in metals and superconductors are given in Refs.~\cite{Slichter1990,tinkham}.
	
	In noncentrosymmetric superconductors, this interplay can acquire a geometric character. Broken inversion symmetry permits antisymmetric spin-orbit coupling (SOC)~\cite{Rashba1960,Dresselhaus}, which generally splits the Fermi surface into helicity sheets and locks spin to momentum. In the zero-field strong-locking regime, the residual Knight-shift anisotropy can then encode the Fermi-surface average of the locking texture rather than being determined by the spin classification of the pairing alone. Zero-field susceptibility has been studied in specific models~\cite{Frigeri2004,Frigeri_2004,Samokhin2007,Edelstein2008} and surveyed in broader reviews~\cite{bauer2012non,Yip14,Samokhin15,Smidman17}. Reference~\cite{pang2026} developed a self-consistent single-band Bogoliubov--de Gennes (BdG) theory of $\chi_{\mu\nu}(T)$ applicable at arbitrary SOC strength and Zeeman field. The companion paper~\cite{pang2025} treats the transition temperature and quasiparticle excitations within the same framework; the present work develops its exact kinematic identities and controlled static and dynamical response consequences.
	
		The analysis distinguishes three logical levels. First, the Bogoliubov sum rule is an exact kinematic identity within a quadratic BdG description. Second, static and dynamical response identities require explicit assumptions about the spectrum, pairing structure, and limiting procedure. Third, extracting conclusions from Knight-shift or relaxation data also requires control of hyperfine coupling, orbital subtraction, and band structure. Maintaining these distinctions prevents an operator-norm budget from being misinterpreted as a universal susceptibility bound.
		
			The central response result belongs to the second level. Consider the zero-field case and assume a fully gapped, helicity-diagonal state with both spin-split helicity sheets crossing the Fermi level. If $\Delta_{\max}\ll g_{\min}$ and $g_{\max}\ll E_F$, where $\Delta_{\max}$ is the maximum gap and $g_{\min}$ and $g_{\max}$ are the minimum and maximum SOC energy scales on the paired Fermi surfaces, then
		\begin{equation}
			\frac{\chi_{\mu\nu}(0)}{\chi_N}
			=\delta_{\mu\nu}-\Pi_{\mu\nu}+o(1),\qquad
			\Pi_{\mu\nu}\equiv
			\left\langle\hat n_{\mathbf k,\mu}\hat n_{\mathbf k,\nu}\right\rangle_{\rm FS}.
			\label{eq:intro-main}
		\end{equation}
		Here $\chi_N$ is the normal-state Pauli susceptibility, $\hat{\mathbf n}_{\mathbf k}$ is the unit spin-locking texture, and $o(1)$ denotes corrections vanishing as $\Delta_{\max}/g_{\min}\to0$ and $g_{\max}/E_F\to0$. The tensor $\Pi$ is symmetric, positive semidefinite, and has unit trace, but is not generally itself a projector. The allowed eigenvalue triples therefore lie on a two-dimensional simplex. The associated \emph{Knight-shift ellipsoid} is defined by the normalized spin response: its semi-axes are the residual principal Knight shifts, and a vanishing response collapses the corresponding axis. This geometry classifies locking textures rather than pairing parity: distinct singlet, triplet, and parity-mixed states can occupy the same simplex point. Here ``geometric'' refers to the Fermi-surface texture of $\hat{\mathbf n}_{\mathbf k}$, not to the quantum-geometric tensor of the Bloch states.
	
		The kinematic starting point for deriving Eq.~\eqref{eq:intro-main} is the exact Bogoliubov sum rule
	\begin{equation}
		\sum_{s_{1}s_{2}}\!\left[
		W^{s_{1}s_{2}}_{ph,O}(\mathbf{k})
		+W^{s_{1}s_{2}}_{pp,O}(\mathbf{k})\right]
		=\mathrm{Tr}_{s}(O^{2}),
		\label{eq:intro-sumrule}
	\end{equation}
			where $W_{ph,O}$ and $W_{pp,O}$ are the sum-rule weights constructed from particle-hole and particle-particle squared matrix elements, as defined in Sec.~\ref{sec:generalized}. This identity holds pointwise at every $\mathbf{k}$ for any Hermitian single-particle operator $O$. It is the BdG-doubled form of Hilbert--Schmidt-norm invariance, $\|U^{\dagger}\mathcal{O}_{\rm BdG}U\|_{F}=\|\mathcal{O}_{\rm BdG}\|_{F}$, with the invariant norm partitioned between the particle-hole and particle-particle blocks. The specialization $O=\sigma_{\mu}$ recovers the matrix-element identity used in Ref.~\cite{pang2026}; other choices generate the charge, multi-orbital, and site-resolved variants used below. Because the sum rule contains neither excitation energies nor occupations, it does not by itself imply Eq.~\eqref{eq:intro-main} or a susceptibility bound; those conclusions require the additional response assumptions stated above.
	
			The scope of these consequences is correspondingly limited. In the strong-locking limit, $\mathrm{Tr}\,\chi(0)=2\chi_N+o(1)$, whereas the zero-field $s$-wave result, exact for arbitrary $|\mathbf g|/\Delta$ within the reference-Fermi-surface regime, gives $\mathrm{Tr}\,\chi(0)\leq2\chi_N$; neither is a universal bound on unitary BdG states. Fully gapped helicity-diagonal parity mixtures approach the same $T=0$ locking tensor but can exhibit distinct finite-temperature recovery scales. Full cubic invariance of both the superconducting state and locking texture fixes $\Pi=\mathbb I/3$ exactly, although its susceptibility consequence remains asymptotic; measurements along cubic directions yield weaker conclusions unless tensor diagonality is independently established. In a finite Zeeman field, the effective locking fields at $\mathbf k$ and $-\mathbf k$ are generally not antiparallel, invalidating the zero-field helicity reduction and requiring a self-consistent BdG calculation~\cite{pang2026}. The $^{75}$As data on K$_2$Cr$_3$As$_3$ by Yang \emph{et al.}~\cite{Triplet2021} accordingly establish a field-dependent axial suppression pattern, not a zero-field tensor. A placement near the $\hat c$-axis vertex and tension with a specified decoupled-pocket SOC-texture baseline follow only under a fixed-tensor $T\to0$ and zero-field extrapolation; a sheet-resolved inference further requires pocket additivity.
	
				The paper follows this hierarchy. Section~\ref{sec:setup} fixes the notation, and Sec.~\ref{sec:generalized} proves the exact Bogoliubov sum rule. Section~\ref{sec:isoT0}, particularly Sec.~\ref{sec:anisoT0}, distinguishes kinematic constraints from response identities and derives the zero-field strong-locking result. Section~\ref{sec:paritymix} treats helicity-diagonal parity mixing, while Sec.~\ref{sec:beyond-locking} treats intermediate SOC and the finite-field numerical regime. Section~\ref{sec:multiband} develops the multiband extension. Section~\ref{sec:ellipsoid} introduces the Knight-shift ellipsoid and its cubic and hyperfine consequences. Section~\ref{sec:guide} translates these results into experimental diagnostics, followed by the conditional K$_2$Cr$_3$As$_3$ application in Sec.~\ref{sec:K2Cr3As3} and the summary in Sec.~\ref{sec:summary}. Appendix~\ref{app:finite-field-response} gives the finite-field definitions and UBe$_{13}$ construction, Appendix~\ref{app:dynamical-response} develops the dynamical-response and $1/T_1$ limitations, and Appendix~\ref{app:intermediate-support} presents the intermediate-SOC derivation and supporting checks.
		
	\section{Setup and notation}
	\label{sec:setup}
	
	We use the spinful single-band BdG framework of Refs.~\cite{pang2025,pang2026,BCS1957,deGennes,tinkham,Slichter1990}. With $\mathbf k$ summed over the full Brillouin zone, the mean-field Hamiltonian is
	\begin{equation}
		H=\tfrac{1}{2}\sum_{\mathbf{k}}C^{\dagger}_{\mathbf{k}}\mathcal{H}_{\rm BdG}(\mathbf{k})C_{\mathbf{k}}+\text{const.},
		\label{eq:HBdG}
	\end{equation}
	where the factor $1/2$ removes Nambu double counting. In the basis
	\[
	C_{\mathbf{k}}=(c_{\mathbf{k}\uparrow},c_{\mathbf{k}\downarrow},
	c^{\dagger}_{-\mathbf{k}\uparrow},c^{\dagger}_{-\mathbf{k}\downarrow})^{T},
	\]
	the BdG matrix has the block form
	\begin{equation}
		\mathcal{H}_{\rm BdG}(\mathbf{k})=\begin{pmatrix}H_{0}(\mathbf{k}) & \Delta(\mathbf{k})\\ \Delta^{\dagger}(\mathbf{k}) & -H_{0}^{T}(-\mathbf{k})\end{pmatrix}.
		\label{eq:BdGblock}
	\end{equation}
	The normal-state block is
	\begin{equation}
		\begin{aligned}
		H_{0}(\mathbf{k})&=\xi_{\mathbf{k}}\sigma_{0}
		+\mu_{B}\mathbf{H}\cdot\hat\sigma
		+\mathbf{g}_{\mathbf{k}}\cdot\hat\sigma,\\
		\xi_{-\mathbf{k}}&=\xi_{\mathbf{k}},\qquad
		\mathbf{g}_{-\mathbf{k}}=-\mathbf{g}_{\mathbf{k}}.
		\end{aligned}
		\label{eq:H0gen}
	\end{equation}
	where $\sigma_{0}$ is the spin identity, $\hat\sigma=(\sigma_x,\sigma_y,\sigma_z)$, and $\xi_{\mathbf{k}}$ is the spin-independent dispersion measured from the chemical potential. The vector $\mathbf{g}_{\mathbf{k}}$ is a general antisymmetric SOC field; no isotropic form such as $g\mathbf{k}$ is assumed~\cite{Rashba1960,Dresselhaus}.

	The most general two-component gap matrix is~\cite{Balian1963,Sigrist1991,Leggett1975}
	\begin{equation}
		\Delta(\mathbf{k})=i\big[\psi(\mathbf{k})\sigma_{0}
		+\mathbf{d}(\mathbf{k})\cdot\hat\sigma\big]\sigma_{y},
		\label{eq:Delta-gen}
	\end{equation}
		Fermionic antisymmetry requires $\psi(-\mathbf{k})=\psi(\mathbf{k})$ and $\mathbf{d}(-\mathbf{k})=-\mathbf{d}(\mathbf{k})$. Broken inversion symmetry permits singlet--triplet mixing but does not require both components; the pure-singlet and pure-triplet limits remain included. The exact Bogoliubov sum rule derived below does not require unitary pairing. When generalized unitarity is invoked below, it means $\Delta\Delta^{\dagger}\propto\mathbb{I}$. It is required for the exact zero-SOC unitary-triplet response identity, but not for the fully gapped helicity-diagonal strong-SOC identities, whose assumptions are stated separately. Suppressing momentum arguments, the generalized-unitarity condition is
	\begin{equation}
		\psi\mathbf d^{*}+\psi^{*}\mathbf d+i\mathbf d\times\mathbf d^{*}=0.
		\label{eq:unitarity-general}
	\end{equation}
	Only in the pure-triplet limit is this equivalent to $\mathbf d\times\mathbf d^{*}=0$.
	
	Introduce quasiparticle operators $\gamma_{\mathbf{k}s}$ through
	\begin{equation}
		c_{\mathbf{k}\alpha}=\sum_{s=\pm}\!\left(
		u^{\alpha s}_{\mathbf{k}}\gamma_{\mathbf{k}s}
		+v^{\alpha s}_{\mathbf{k}}\gamma^{\dagger}_{-\mathbf{k}s}\right),
		\label{eq:BdGtrans}
	\end{equation}
	where $\alpha$ is a spin index and $s=\pm$ labels the two positive-energy branches, $E_{\mathbf{k}s}\ge 0$. In Nambu form,
	\begin{equation}
		C_{\mathbf{k}}=U_{\mathbf{k}}\Gamma_{\mathbf{k}},\qquad
		U_{\mathbf{k}}=\begin{pmatrix}u_{\mathbf{k}} & v_{\mathbf{k}}\\v^{*}_{-\mathbf{k}} & u^{*}_{-\mathbf{k}}\end{pmatrix},\qquad
		U^{\dagger}_{\mathbf{k}}U_{\mathbf{k}}=\mathbb{I}_{\rm BdG},
		\label{eq:Umatrix}
	\end{equation}
	with
	\[
	\Gamma_{\mathbf{k}}=(\gamma_{\mathbf{k}+},\gamma_{\mathbf{k}-},
	\gamma^{\dagger}_{-\mathbf{k}+},\gamma^{\dagger}_{-\mathbf{k}-})^{T}.
	\]
	In the response notation of Ref.~\cite{pang2026}, $ph$ and $pp$ denote particle-hole and particle-particle channels, while $+$ and $-$ denote intra- and interbranch contributions, respectively. Thus the diagonal static uniform susceptibility is a sum of four non-negative terms,
	\[
	\chi_{\mu\mu}=\chi^{ph+}_{\mu\mu}+\chi^{ph-}_{\mu\mu}
	+\chi^{pp+}_{\mu\mu}+\chi^{pp-}_{\mu\mu}.
	\]
	Throughout, $\mathrm{Tr}_{s}$ denotes the trace over the internal single-particle space and $\mathrm{Tr}_{\rm BdG}$ the trace over its Nambu doubling.
	
	The effective normal-state field and its local spin axis are
	\begin{equation}
		\mathbf{b}_{\mathbf{k}}\equiv\mu_{B}\mathbf{H}+\mathbf{g}_{\mathbf{k}},
		\qquad
		\hat{\mathbf n}_{\mathbf{k}}=
		\frac{\mathbf{b}_{\mathbf{k}}}{|\mathbf{b}_{\mathbf{k}}|},
		\label{eq:ndef}
	\end{equation}
		We assume that $|\mathbf{b}_{\mathbf{k}}|\neq0$ almost everywhere on the reference Fermi surface. The normal-state helicity branch $\lambda=\pm$ has spin expectation parallel to $\lambda\hat{\mathbf n}_{\mathbf{k}}$. At zero Zeeman field, $\hat{\mathbf n}_{\mathbf{k}}=\hat{\mathbf g}_{\mathbf{k}}$.

	We define the reference-Fermi-surface average by
	\[
	\langle X\rangle_{\rm FS}\equiv
	\frac{1}{N(0)}\sum_{\mathbf{k}}X(\mathbf{k})\delta(\xi_{\mathbf{k}}),
	\qquad
	N(0)\equiv\sum_{\mathbf{k}}\delta(\xi_{\mathbf{k}}),
	\]
	where $N(0)$ is the single-spin density of states and $\chi_{N}=2\mu_{B}^{2}N(0)$ is the reference Pauli susceptibility~\cite{Yosida1958,van1928dielectric}. The normal-state texture tensor is
	\begin{equation}
			\Pi_{\mu\nu}(\mathbf{H},\mathbf{g})\equiv\big\langle\hat n_{\mu}(\mathbf{k})\,\hat n_{\nu}(\mathbf{k})\big\rangle_{\rm FS}.
		\label{eq:Pidef}
	\end{equation}
	It is symmetric and positive semidefinite, with $\mathrm{Tr}\,\Pi=1$, and is generally not idempotent because it averages rank-one projectors. In the zero-field strong-locking regime, this tensor controls the residual response. At finite field, it remains a texture descriptor, but it cannot generally be substituted into the zero-field response identity. Table~\ref{tab:notation} in Appendix~\ref{app:notation} collects the notation.
	
\section[Exact Identity]{Exact identity: the Bogoliubov sum rule}
\label{sec:generalized}

The spin matrix-element identity used in Ref.~\cite{pang2026} is a special case of an exact Bogoliubov sum rule for any Hermitian single-particle operator. Its proof uses only Hilbert--Schmidt-norm invariance under the canonical transformation $U_{\mathbf{k}}$. The general form generates spin, charge, multi-orbital, and site-resolved operator budgets; physical responses require additional dynamical assumptions.

\subsection{Nambu embedding of a single-particle observable}

	Let $O=O^{\dagger}$ act on the finite-dimensional internal single-particle space, including spin and, where present, orbital or site indices. For a momentum-independent kernel, the physical one-body observable is $\hat O=\sum_{\mathbf{k}}c^{\dagger}_{\mathbf{k}}Oc_{\mathbf{k}}$. Because $C_{\mathbf{k}}$ combines the $\mathbf{k}$ particle sector with the $-\mathbf{k}$ hole sector, define the even pair contribution
	\begin{equation}
	\hat{O}^{[\mathbf{k}]}\equiv\tfrac{1}{2}\left(
	c^{\dagger}_{\mathbf{k}}Oc_{\mathbf{k}}
	+c^{\dagger}_{-\mathbf{k}}Oc_{-\mathbf{k}}\right).
	\label{eq:Opair}
	\end{equation}
	It has the exact Nambu representation
	\begin{equation}
	\begin{split}
	\hat{O}^{[\mathbf{k}]}&=C^{\dagger}_{\mathbf{k}}\mathcal{O}_{\rm BdG}C_{\mathbf{k}}+\tfrac{1}{2}\mathrm{Tr}_{s}O,\\
	\mathcal{O}_{\rm BdG}&=\tfrac{1}{2}\begin{pmatrix}O & 0\\0 & -O^{T}\end{pmatrix}.
	\end{split}
	\label{eq:Oembedding}
	\end{equation}
	The transpose and minus sign in the hole block follow from fermionic reordering. For Hermitian $O$, $\mathcal{O}_{\rm BdG}$ is Hermitian, and the additive term is a c-number that drops out of connected correlation functions and response matrix elements. Moreover,
	\begin{equation}
	\mathrm{Tr}_{\rm BdG}\mathcal{O}_{\rm BdG}=0,\qquad
	\mathrm{Tr}_{\rm BdG}\mathcal{O}_{\rm BdG}^{2}
	=\tfrac{1}{2}\mathrm{Tr}_{s}O^{2}.
	\label{eq:BdGtrace}
	\end{equation}
	Since $\hat O^{[\mathbf{k}]}=\hat O^{[-\mathbf{k}]}$, summing it over the full Brillouin zone recovers $\hat O$. Equivalently, one may use one representative of each generic pair with weight two, while counting self-inverse momenta $\mathbf{k}\equiv-\mathbf{k}$ modulo a reciprocal vector only once.
	
	\subsection{The Bogoliubov sum rule}
	
	In the quasiparticle basis, define $\tilde{\mathcal{O}}_{\mathbf{k}}=U^{\dagger}_{\mathbf{k}}\mathcal{O}_{\rm BdG}U_{\mathbf{k}}$ and write its Nambu block decomposition as
	\begin{equation}
	2\tilde{\mathcal{O}}_{\mathbf{k}}
	=\begin{pmatrix}
	A_{O}(\mathbf{k}) & B_{O}(\mathbf{k})\\
	B_{O}^{\dagger}(\mathbf{k}) & D_{O}(\mathbf{k})
	\end{pmatrix}
	\equiv
	\begin{pmatrix}
	M_{ph,O}(\mathbf{k}) & M_{pp,O}(\mathbf{k})\\
	M^{\dagger}_{pp,O}(\mathbf{k}) & \widetilde{M}_{ph,O}(\mathbf{k})
	\end{pmatrix}.
	\label{eq:block}
	\end{equation}
	The diagonal blocks contain particle-hole matrix elements, whereas the off-diagonal blocks contain particle-particle matrix elements. Explicitly,
	\begin{align}
	M_{ph,O}(\mathbf{k}) & =u^{\dagger}_{\mathbf{k}}Ou_{\mathbf{k}}-[v^{\dagger}_{-\mathbf{k}}Ov_{-\mathbf{k}}]^{T},\label{eq:Mph-O}\\
	M_{pp,O}(\mathbf{k}) & =u^{\dagger}_{\mathbf{k}}Ov_{\mathbf{k}}-[u^{\dagger}_{-\mathbf{k}}Ov_{-\mathbf{k}}]^{T}.\label{eq:Mpp-O}
	\end{align}
	The lower diagonal block satisfies
	\[
	D_{O}(\mathbf{k})=\widetilde{M}_{ph,O}(\mathbf{k})
	=-M_{ph,O}(-\mathbf{k})^{T}.
	\]
	Superscripts $s_{1}s_{2}$ denote matrix components in the positive-energy quasiparticle-branch basis. They represent helicity indices only in regimes where the quasiparticle branches can be identified with normal-state helicities.

		\begin{theorem}
	Let $O=O^{\dagger}$ act on a finite-dimensional internal single-particle space, and let $U_{\mathbf{k}}$ be a unitary canonical Bogoliubov transformation. For the blocks in Eq.~\eqref{eq:block},
	\begin{equation}
	\tfrac{1}{2}\!\left(\|A_{O}(\mathbf{k})\|_{F}^{2}
	\mathbin{+}\|D_{O}(\mathbf{k})\|_{F}^{2}\right)
	\mathbin{+}\|B_{O}(\mathbf{k})\|_{F}^{2}
	=\mathrm{Tr}_{s}(O^{2}),
	\label{eq:bogoliubov-HS}
	\end{equation}
	where $\|X\|_{F}^{2}\equiv\mathrm{Tr}(X^{\dagger}X)$.
	\end{theorem}

	\emph{Proof.}
	Hilbert--Schmidt-norm invariance and Eq.~\eqref{eq:BdGtrace} give
	\begin{equation}
	\|2\tilde{\mathcal O}_{\mathbf{k}}\|_{F}^{2}
	=4\|\mathcal O_{\rm BdG}\|_{F}^{2}
	=2\,\mathrm{Tr}_{s}(O^{2}).
	\label{eq:HS-invariance}
	\end{equation}
	From Eq.~\eqref{eq:block},
	\[
	\|2\tilde{\mathcal O}_{\mathbf{k}}\|_{F}^{2}
	=\|A_{O}\|_{F}^{2}+\|D_{O}\|_{F}^{2}+2\|B_{O}\|_{F}^{2}.
	\]
	Equating the two expressions and dividing by two proves the result. \hfill$\square$

		Only unitarity of $U_{\mathbf{k}}$ enters the norm proof; its canonical BdG structure identifies the diagonal and off-diagonal blocks as particle-hole and particle-particle sectors. The identity contains no quasiparticle energies, occupations, temperature, pairing data, or Hamiltonian energy scales. Individual matrix elements and the partitioning between sectors are basis dependent, although block norms are invariant under unitary rotations within fixed particle or hole subspaces. It is therefore an operator budget, not a response identity.

		\begin{corollary}[Component form]
		\label{cor:component}
		Using $D_{O}(\mathbf{k})=-M_{ph,O}(-\mathbf{k})^{T}$ and invariance of the Frobenius norm under transposition, Eq.~\eqref{eq:bogoliubov-HS} becomes
	\begin{equation}
	\begin{aligned}
	&\tfrac{1}{2}\sum_{s_{1}s_{2}}\!\left[
	|M^{s_{1}s_{2}}_{ph,O}(\mathbf{k})|^{2}
	\mathbin{+}|M^{s_{1}s_{2}}_{ph,O}(-\mathbf{k})|^{2}\right]\\
	&\qquad
	\mathbin{+}\sum_{s_{1}s_{2}}
	|M^{s_{1}s_{2}}_{pp,O}(\mathbf{k})|^{2}
	=\mathrm{Tr}_{s}(O^{2}).
	\end{aligned}
	\label{eq:bogoliubov-component}
	\end{equation}
	Define the non-negative \emph{sum-rule weights}
	\begin{align}
	W^{s_{1}s_{2}}_{ph,O}(\mathbf{k})&\equiv\tfrac{1}{2}\big[|M^{s_{1}s_{2}}_{ph,O}(\mathbf{k})|^{2}+|M^{s_{1}s_{2}}_{ph,O}(-\mathbf{k})|^{2}\big],\label{eq:Wph-def}\\
	W^{s_{1}s_{2}}_{pp,O}(\mathbf{k})&\equiv|M^{s_{1}s_{2}}_{pp,O}(\mathbf{k})|^{2},\label{eq:Wpp-def}
	\end{align}
	Then
	\begin{equation}
	\boxed{\;\sum_{s_{1}s_{2}}\!\big[W^{s_{1}s_{2}}_{ph,O}(\mathbf{k})+W^{s_{1}s_{2}}_{pp,O}(\mathbf{k})\big]=\mathrm{Tr}_{s}(O^{2}),\;}
		\label{eq:general-sum-rule}
		\end{equation}
		The particle-hole weight is explicitly invariant under $\mathbf{k}\leftrightarrow-\mathbf{k}$. Individual particle-particle components need not be invariant, but their total weight is the same at the paired momenta. At a self-inverse momentum, the particle-hole average reduces to a single squared matrix element.
		\end{corollary}

	For a momentum-dependent Hermitian kernel $O(\mathbf{k})=O(\mathbf{k})^{\dagger}$, define the paired contribution using $O(\mathbf{k})$ and $O(-\mathbf{k})$. Its Nambu kernel is
	\[
	\mathcal O_{\rm BdG}(\mathbf{k})
	=\tfrac12\operatorname{diag}[O(\mathbf{k}),-O(-\mathbf{k})^{T}].
	\]
	The same norm argument replaces the right-hand side of Eq.~\eqref{eq:bogoliubov-HS} by
	\[
	\tfrac12\left\{
	\mathrm{Tr}_{s}[O(\mathbf{k})^{2}]
	+\mathrm{Tr}_{s}[O(-\mathbf{k})^{2}]\right\}.
	\]
	If these two traces are equal, as when $O(-\mathbf{k})=\pm O(\mathbf{k})$, the result reduces to $\mathrm{Tr}_{s}[O(\mathbf{k})^{2}]$.
		
		\emph{Relation to Ref.~\cite{pang2026}.} Equation~(14) of Ref.~\cite{pang2026} is recovered by setting $O=\sigma_{\mu}$ in Eq.~\eqref{eq:general-sum-rule}.

\subsection{Specializations}
\label{sec:specializations}

(a) \emph{Spin sector.} For each $O=\sigma_{\mu}$,
\[
\sum_{s_1s_2}\left(W_{ph,\mu}^{s_1s_2}
+W_{pp,\mu}^{s_1s_2}\right)=2,
\]
and summing over $\mu=x,y,z$ gives
\begin{equation}
\sum_{\mu}\sum_{s_{1}s_{2}}\big[W^{s_{1}s_{2}}_{ph,\mu}(\mathbf{k})+W^{s_{1}s_{2}}_{pp,\mu}(\mathbf{k})\big]=6.
\label{eq:spin-total}
\end{equation}

(b) \emph{Scalar-density sector.} For $O=\sigma_{0}$,
\[
\sum_{s_1s_2}\left(W_{ph,0}^{s_1s_2}
+W_{pp,0}^{s_1s_2}\right)=2.
\]
This is the momentum-diagonal scalar-density budget. A finite-$\mathbf{q}$ response additionally requires vertices connecting $\mathbf{k}$ and $\mathbf{k}+\mathbf{q}$, energy denominators, and gauge constraints.

(c) \emph{Multi-orbital sector.} For a Hermitian orbital matrix $\Lambda$,
\[
O=\sigma_{\mu}\otimes\Lambda,\qquad
\mathrm{Tr}_{s}(O^{2})=2\,\mathrm{Tr}_{\rm orb}(\Lambda^{2}).
\]
In particular, $\Lambda=\mathbb{I}_{\rm orb}$ gives $2N_{\rm orb}$. This is an exact operator budget in the full internal space. An additive multi-pocket response additionally requires the decoupled-pocket assumptions stated in Sec.~\ref{sec:multiband}.

(d) \emph{Site-resolved hyperfine sector.} For a real hyperfine tensor, define
\[
O_{{\rm hf},\alpha}^{(\mathbf{R})}
\equiv\sum_{\mu}A^{(\mathbf{R})}_{\alpha\mu}\sigma_{\mu}.
\]
Pauli algebra and Eq.~\eqref{eq:general-sum-rule} give
\begin{equation}
\begin{aligned}
\sum_{s_1s_2}\!\left(W_{ph,O_{{\rm hf},\alpha}}^{s_1s_2}
+W_{pp,O_{{\rm hf},\alpha}}^{s_1s_2}\right)
&=2\sum_{\mu}\left(A^{(\mathbf{R})}_{\alpha\mu}\right)^{2},\\
\sum_{\alpha s_1s_2}\!\left(W_{ph,O_{{\rm hf},\alpha}}^{s_1s_2}
+W_{pp,O_{{\rm hf},\alpha}}^{s_1s_2}\right)
&=2\left\|A^{(\mathbf{R})}\right\|_{F}^{2}.
\end{aligned}
\label{eq:site-budget}
\end{equation}
This exact componentwise budget is quadratic in $A$. It does not imply a universal trace ratio for the measured Knight-shift tensor $K=A\chi$, which is linear in $A$. Site-independent normalized simplex coordinates follow only for scalar or coaxial diagonal hyperfine tensors; a general non-coaxial tensor must be determined independently.

\begin{table*}[tb]
		\caption{Hermitian single-particle operators used in Sec.~\ref{sec:specializations}, their $\mathrm{Tr}_{s}(O^{2})$ in the Bogoliubov sum rule [Eq.~\eqref{eq:general-sum-rule}], and their operator-budget roles. These entries constrain matrix-element weights, not physical responses without further assumptions.}
		\label{tab:operators}
		\small
		\renewcommand\arraystretch{2.0}
		\setlength{\tabcolsep}{4.0ex}
		\begin{tabular*}{\textwidth}{@{\extracolsep{\fill}} l l c l @{}}
			\hline\hline
			Operator(s) & Specialization & $\mathrm{Tr}_{s}(O^{2})$ & Budget/role\\
			\hline
			$\sigma_{\mu}$ & Spin component & $2$ & Componentwise spin budget\\
			$\{\sigma_{\mu}\}_{\mu=x,y,z}$ & Total spin & $\sum_\mu\mathrm{Tr}_s(\sigma_\mu^2)=6$ & Summed spin-weight budget\\
			$\sigma_{0}$ & Scalar density & $2$ & Momentum-diagonal density budget\\
			$O_{{\rm hf},\alpha}^{(\mathbf{R})}$ & Site-$\mathbf{R}$ hyperfine component & $2\sum_{\mu}(A^{(\mathbf{R})}_{\alpha\mu})^{2}$ & Exact site-resolved budget\\
			$\sigma_{\mu}\otimes\mathbb{I}_{\rm orb}$ & Full-space spin operator & $2N_{\rm orb}$ & Full-space operator budget\\
			$\sigma_{\mu}\otimes\Lambda$ & Orbital-weighted spin & $2\,\mathrm{Tr}(\Lambda^{2})$ & Orbital-weighted operator budget\\
			\hline\hline
		\end{tabular*}
\end{table*}

The spin sector drives the remainder of the paper; the other specializations define its scope and qualify the experimental protocols and decoupled-pocket baseline.

\subsection{$T=0$ constraint on static uniform response}
\label{sec:T0-consequence}

At $T=0$, the particle-hole channel vanishes for the static uniform susceptibility generated by $\hat O$ whenever there is no residual zero-energy quasiparticle contribution, in particular for a fully gapped spectrum~\cite{pang2026}. For a source field $h$ coupled as $-h\hat O$, the remaining response is
\begin{equation}
\chi_{O}(0)=\sum_{\mathbf{k}}\sum_{s_{1}s_{2}}\frac{W^{s_{1}s_{2}}_{pp,O}(\mathbf{k})}{E_{\mathbf{k}s_{1}}+E_{-\mathbf{k}s_{2}}}.
\label{eq:chiO-T0}
\end{equation}
For magnetic spin response, the corresponding factor $\mu_B^2$ is restored.
Equation~\eqref{eq:general-sum-rule} bounds the numerator pointwise,
\begin{equation}
\sum_{s_{1}s_{2}}W^{s_{1}s_{2}}_{pp,O}(\mathbf{k})\le\mathrm{Tr}_{s}(O^{2}),
\label{eq:p-pointwise}
\end{equation}
for every Hermitian $O$ and every $\mathbf{k}$. This is the strongest response consequence without additional spectral assumptions. Because the denominators in Eq.~\eqref{eq:chiO-T0} depend on the quasiparticle spectrum and the normal-state reference has a different channel decomposition, the numerator bound implies neither $\chi_{\mu\mu}(0)\le2\chi_N/3$ nor a universal trace bound.

This distinction is explicit in the intermediate-SOC calculations of Ref.~\cite{pang2026}: unitary opposite-spin-paired states can have a transverse component exceeding $\chi_N$ while the longitudinal response remains nonzero, giving $\mathrm{Tr}\,\chi(0)>2\chi_N$ within the stated BdG framework. Deviations on either side of $2\chi_N$ therefore do not uniquely diagnose non-unitary pairing. First-order Pauli-limited transitions~\cite{Maki1964,Clogston1962,Chandrasekhar62}, FFLO states~\cite{FF_state,LO_state}, pair-density waves~\cite{PDW}, and correlations beyond mean-field theory require separate analyses.

Accordingly, each susceptibility identity below is introduced separately with its spectral, pairing, and normalization assumptions.

	\section[Controlled Responses]{Response identities in controlled limits}
	\label{sec:isoT0}
	
	\subsection{Benchmark limits and normal-state response}
	
		Table~\ref{tab:sat} collects $\mathrm{Tr}\,\chi(0)/\chi_{N}$ in controlled limiting cases reported in Ref.~\cite{pang2026}. The trace value $2$ is the leading zero-field strong-locking limit and an exact $T=0$ identity for fully gapped zero-SOC unitary-triplet states; it is not saturation of a universal upper bound.
	
	\begin{table*}[tb]
				\caption{Residual susceptibility in controlled limits, in units of $\chi_N$. The zero-SOC triplet rows assume full gaps; the ESP row also assumes an in-plane-isotropic $\mathbf d$ texture. Strong-SOC rows give leading zero-field values for full gaps, $\Delta_{\max}\ll g_{\min}$, $g_{\max}\ll E_F$, and isotropic locking $\Pi=\mathbb I/3$. Strong locking alone fixes the trace, not individual components. Simplex mapping applies only in the stated locking limits.}
		\label{tab:sat}
		\small
		\renewcommand\arraystretch{2.0}
		\setlength{\tabcolsep}{2.0ex}
			\begin{tabular*}{\textwidth}{@{\extracolsep{\fill}} l c c c c l l @{}}
			\hline\hline
			Pairing/limit & $\chi_{xx}(0)$ & $\chi_{yy}(0)$ & $\chi_{zz}(0)$ & Tr$/\chi_{N}$ & Locking $\hat{\mathbf n}_{\mathbf{k}}$ & Mechanism \\
			\hline
			$s$-wave, $\mathbf{g}=\mathbf{H}=0$ & $0$ & $0$ & $0$ & $0$ & --- & No locking (Yosida)\\
			$s$-wave, isotropic strong SOC & $2/3$ & $2/3$ & $2/3$ & $2$ & $\hat{\mathbf g}_{\mathbf{k}}$ & SOC ($\langle\hat g_{\mu}^{2}\rangle=1/3$)\\
			OSP ($\mathbf{d}\parallel\hat z$), $\mathbf{g}=\mathbf{H}=0$ & $1$ & $1$ & $0$ & $2$ & $\hat z$ & $\mathbf{d}$-vector\\
			ESP ($\mathbf{d}\in xy$), $\mathbf{g}=\mathbf{H}=0$ & $1/2$ & $1/2$ & $1$ & $2$ & $\hat d_{\mathbf{k}}\in xy$ & $\mathbf{d}$-vector\\
			Unitary $p$-wave, isotropic strong SOC & $2/3$ & $2/3$ & $2/3$ & $2$ & $\hat{\mathbf g}_{\mathbf{k}}$ & SOC dominates\\
			\hline\hline
			\end{tabular*}
	\end{table*}
	
	The table compares distinct controlled limits and must not be interpolated into a general inequality. Section~\ref{sec:intermediate} gives one exact interpolation for the zero-field $s$-wave problem; other pairing states require their own calculation.
	
			At $\mathbf H=0$ and $T_{c}^{+}$, the intraband and interband ph weights for $O=\sigma_{\mu}$ are
		\begin{align}
			W^{\rm intra}_{ph,\mu}(\mathbf{k})\equiv\sum_{s}W^{ss}_{ph,\mu}(\mathbf{k})&=2(\hat{\mathbf n}_{\mathbf{k}}\cdot\hat e_{\mu})^{2},\label{eq:me-intra}\\
			W^{\rm inter}_{ph,\mu}(\mathbf{k})\equiv\sum_{s\ne s'}W^{ss'}_{ph,\mu}(\mathbf{k})&=2[1-(\hat{\mathbf n}_{\mathbf{k}}\cdot\hat e_{\mu})^{2}],\label{eq:me-inter}
		\end{align}
			with the interband sum including both off-diagonal elements. Assume that both helicity sheets cross the Fermi level, and define $g_{\max}\equiv\max_{\mathbf k\in{\rm FS}}|\mathbf g_{\mathbf k}|$. In the reference-Fermi-surface limit $g_{\max}\ll E_F$, inserting these weights into Eqs.~(13c)--(13d) of Ref.~\cite{pang2026}, using $-df/dE\to\delta(E)$ for the intraband term and the corresponding small-splitting limit of the interband difference quotient, gives
		\begin{equation}
			\chi^{ph+}_{\mu\mu}(T_{c})/\chi_{N}=\Pi_{\mu\mu},\quad \chi^{ph-}_{\mu\mu}(T_{c})/\chi_{N}=1-\Pi_{\mu\mu}.
			\label{eq:phpartsTc}
		\end{equation}
		The two contributions add to $\chi^{ph}_{\mu\mu}(T_{c})=\chi_{N}$, so $\mathrm{Tr}\,\chi(T_{c})=3\chi_{N}$ at this order for any reference FS and locking texture. This is a weak-band-splitting normal-state benchmark, not an arbitrary-SOC statement.
	
		\subsection{Strong-SOC locking identity at $T=0$}
		\label{sec:anisoT0}
		
			We now impose the assumptions needed for a controlled response identity: $\mathbf H=0$, a full gap, both helicity sheets crossing the Fermi level, negligible transverse inter-helicity pairing, and
			\begin{equation}
				\Delta_{\max}\ll g_{\min},\qquad g_{\max}\ll E_F,\qquad
				g_{\min}\equiv\min_{\mathbf k\in{\rm FS}}|\mathbf g_{\mathbf k}|.
				\label{eq:strong-locking-scales}
			\end{equation}
				The first inequality ensures that the SOC energy scale dominates pairing, while the second preserves the two-sheet reference-Fermi-surface description. We derive the leading response in three steps: helicity block-diagonalization, longitudinal Yosida cancellation in the intra-helicity sector, and evaluation of the inter-helicity pp contribution.
	
			(i) \emph{Helicity block-diagonalization.} In the helicity basis obtained by rotating at each $\mathbf{k}$ to diagonalize $H_{0}(\mathbf{k})$, a helicity-diagonal gap has amplitudes $\Delta_\pm(\mathbf k)$ and $\Delta_\perp=0$. To leading order in $\Delta_{\max}/g_{\min}$, the BdG Hamiltonian decouples into two intra-helicity $2\times2$ BCS blocks
	\begin{equation}
		H^{\lambda}_{\rm BdG}(\mathbf{k})=\begin{pmatrix}\xi_{\lambda}(\mathbf{k}) & \Delta_{\lambda}(\mathbf{k})\\ \Delta^{*}_{\lambda}(\mathbf{k}) & -\xi_{\lambda}(\mathbf{k})\end{pmatrix}, \quad \xi_{\lambda}\equiv\xi_{\mathbf{k}}+\lambda|\mathbf{g}_{\mathbf{k}}|,
		\label{eq:helicity-block}
	\end{equation}
		with $E_{\lambda}=\sqrt{\xi_{\lambda}^{2}+|\Delta_{\lambda}|^{2}}$ and standard BCS coherence-factor magnitudes $|u_{\lambda}|^{2}=(1+\xi_{\lambda}/E_{\lambda})/2$, $|v_{\lambda}|^{2}=(1-\xi_{\lambda}/E_{\lambda})/2$; the relative phase is fixed by $\Delta_{\lambda}$. Each block describes pairing of the helicity Kramers pair $\{|\lambda,\mathbf{k}\rangle,|\lambda,-\mathbf{k}\rangle\}$, whose spin expectations are $\lambda\hat{\mathbf n}_{\mathbf{k}}$ and $\lambda\hat{\mathbf n}_{-\mathbf{k}}=-\lambda\hat{\mathbf n}_{\mathbf{k}}$, respectively, for antisymmetric SOC.
	
		(ii) \emph{Intra-helicity sector: longitudinal Yosida cancellation.} Within each helicity sheet, the spin operator $\sigma_{\mu}$ projects onto the local pseudo-spin quantization axis $\hat{\mathbf n}_{\mathbf{k}}$ as a longitudinal operator with a weight $(\hat{\mathbf n}_{\mathbf{k}}\cdot\hat e_{\mu})$; the component transverse to $\hat{\mathbf n}_{\mathbf{k}}$ connects the two helicity sheets and contributes only to the inter-helicity sector below. The intra-helicity contribution to $\chi_{\mu\mu}$ is therefore the longitudinal-Yosida response of a fully gapped BCS condensate,
	\begin{equation}
		\chi^{\rm intra}_{\mu\mu}(T)=\chi_{N}\sum_{\lambda}w_{\lambda}\big\langle(\hat{\mathbf n}_{\mathbf{k}}\cdot\hat e_{\mu})^{2}\,Y_{\lambda}(\mathbf{k},T)\big\rangle_{\lambda},
		\label{eq:intra-Yosida}
	\end{equation}
	with $w_{\lambda}\equiv N_{\lambda}(0)/\sum_{\lambda'=\pm}N_{\lambda'}(0)$ the normalized helicity-sheet DOS fraction, $\sum_{\lambda}w_{\lambda}=1$, and $Y_{\lambda}$ the helicity-resolved Yosida function. For a fully gapped state $Y_{\lambda}(\mathbf{k},0)=0$, so
	\begin{equation}
		\chi^{\rm intra}_{\mu\mu}(0)=0.
		\label{eq:intra-zero}
	\end{equation}
		The vanishing is exact at this order, not parametrically small in $\Delta/|\mathbf{g}|$: it is the standard $T=0$ suppression of the longitudinal spin response on a fully polarized BCS sheet.
	
	(iii) \emph{Inter-helicity sector: coherence factor at strong locking.} The inter-helicity pp contribution to $\chi_{\mu\mu}$ at $T=0$ is
	\begin{equation}
		\chi^{\rm inter}_{\mu\mu}(0)=\mu_{B}^{2}\sum_{\mathbf{k}}\sum_{\lambda\ne\lambda'}\frac{W^{\lambda\lambda'}_{pp,\mu}(\mathbf{k})}{E_{\mathbf{k}\lambda}+E_{-\mathbf{k}\lambda'}}.
		\label{eq:inter-pp-chi}
	\end{equation}
	The pp matrix element factorizes into a spin-matrix-element factor and a Bogoliubov coherence factor,
	\begin{equation}
		W^{\lambda\lambda'}_{pp,\mu}(\mathbf{k})=\big[1-(\hat{\mathbf n}_{\mathbf{k}}\cdot\hat e_{\mu})^{2}\big]\,\mathcal{C}^{\lambda\lambda'}(\mathbf{k}),
		\label{eq:Wpp-factored}
	\end{equation}
	with the inter-helicity coherence factor
	\begin{equation}
		\mathcal{C}^{\lambda\lambda'}(\mathbf{k})\equiv\big|u_{\lambda}(\mathbf{k})v_{\lambda'}(\mathbf{k})-u_{\lambda'}(\mathbf{k})v_{\lambda}(\mathbf{k})\big|^{2}.
		\label{eq:Cdef}
	\end{equation}
			For antisymmetric SOC, $|\mathbf{g}_{-\mathbf{k}}|=|\mathbf{g}_{\mathbf{k}}|$ implies $\xi_{\lambda}(-\mathbf{k})=\xi_{\lambda}(\mathbf{k})$, so $u_{\lambda},v_{\lambda}$ are even in $\mathbf{k}$. For illustration, on the $+$ FS one has $\xi_{+}=0$ and $\xi_{-}=-2|\mathbf{g}_{\mathbf{k}}|$:
			\begin{align}
				|u_{+}(\mathbf{k})|&=|v_{+}(\mathbf{k})|=\tfrac{1}{\sqrt{2}},\\
				|u_{-}(\mathbf{k})|&=O(\Delta_{\max}/|\mathbf{g}_{\mathbf{k}}|),\; |v_{-}(\mathbf{k})|=1+O(\Delta_{\max}^{2}/|\mathbf{g}_{\mathbf{k}}|^{2}),\\
				\mathcal{C}^{+-}(\mathbf{k})&=\tfrac{1}{2}+O(\Delta_{\max}/|\mathbf{g}_{\mathbf{k}}|).\label{eq:C-evaluated}
			\end{align}
			The symmetry of Eq.~\eqref{eq:Cdef} under $\lambda\leftrightarrow\lambda'$ gives $\mathcal{C}^{-+}(\mathbf{k})=\mathcal{C}^{+-}(\mathbf{k})$, and the analogous result holds on the $-$ FS by $+\leftrightarrow-$. The denominator on either FS is $E_{\mathbf{k}\lambda}+E_{-\mathbf{k}\lambda'}=2|\mathbf{g}_{\mathbf{k}}|+O(\Delta_{\max})$.
		
				The sheetwise expansion is not uniform over the entire $\xi$ integration, whose range extends from the gap scale to the inter-helicity splitting. After the $\xi$ integration, the strong-locking limit recovers the transverse matrix-element factor of the normal-state benchmark Eq.~\eqref{eq:phpartsTc}:
			\begin{equation}
				\chi^{\rm inter}_{\mu\mu}(0)=\chi_{N}\,(1-\Pi_{\mu\mu})
				+o(\chi_N),\qquad \frac{\Delta_{\max}}{g_{\min}}\to0.
				\label{eq:inter-result}
		\end{equation}
		For the zero-field $s$-wave problem, the exact result of Sec.~\ref{sec:intermediate} sharpens the remainder to $O[\chi_N(\Delta_{\max}/g_{\min})^2\ln(g_{\min}/\Delta_{\max})]$.
		Thus the inter-helicity pp response approaches the normal-state inter-helicity Van Vleck response as $\Delta_{\max}/g_{\min}\to0$.
	
		The same derivation extends to off-diagonal $\chi_{\mu\nu}$ by replacing the diagonal spin-matrix-element factor $1-(\hat{\mathbf n}_{\mathbf{k}}\cdot\hat e_{\mu})^{2}$ in Eq.~\eqref{eq:Wpp-factored} with $\delta_{\mu\nu}-\hat n_{\mu}(\mathbf{k})\hat n_{\nu}(\mathbf{k})$, the only modification being which Cartesian components of $\hat\sigma$ enter the matrix element. Combining Eqs.~\eqref{eq:intra-zero} and~\eqref{eq:inter-result} (and its off-diagonal analogue) gives
			\begin{equation}
				\frac{\chi_{\mu\nu}(0)}{\chi_N}
				=\delta_{\mu\nu}-\Pi_{\mu\nu}+o(1),\qquad
				\frac{\Delta_{\max}}{g_{\min}}\to0,\quad
				\frac{g_{\max}}{E_F}\to0,
				\label{eq:strongSOCid}
			\end{equation}
			Equation~\eqref{eq:Fs} supplies the explicit zero-field $s$-wave coefficient, whose strong-locking expansion $F_s(\lambda)=1-\ln(2\lambda)/\lambda^2+\cdots$ is given in Appendix~\ref{app:intermediate-support}. Other fully gapped helicity-diagonal states have the same leading tensor structure, while the detailed subleading term depends on their gaps and Fermi-surface weights.
		
		At leading order, $\mathrm{Tr}\,\chi(0)=2\chi_N$, tilted textures can generate off-diagonal $\chi_{\mu\nu}(0)$, and the eigenvectors of $\Pi$ define the spin-locking principal axes. These are strong-locking consequences, not universal BdG bounds.
		
			\emph{Off-diagonal texture example.} For disconnected quasi-1D sheets with $|k_{\parallel}|\simeq k_F$, let $\mathbf g_{\mathbf k}=k_{\parallel}(\alpha\hat x+\beta\hat z)$. Then $\Pi_{xz}=\alpha\beta/(\alpha^2+\beta^2)\ne0$, producing an off-diagonal electronic response $\chi_{xz}$. This response appears as $K_{xz}$ only for scalar or independently calibrated hyperfine coupling.
		
		\subsection{Zero-SOC unitary-triplet analogue and unified form}
		
				At zero SOC, a fully gapped unitary-triplet state~\cite{Leggett1975,Sigrist1991} obeys the analogous $T=0$ identity
		\begin{equation}
			\chi_{\mu\nu}(0)=\chi_{N}[\delta_{\mu\nu}-\Pi^{(d)}_{\mu\nu}],\quad
			\Pi^{(d)}_{\mu\nu}
			=\left\langle\frac{d_\mu d_\nu^{*}}{|\mathbf d|^2}\right\rangle_{\rm FS}.
			\label{eq:strong-triplet}
		\end{equation}
		Unitarity, $\mathbf d\times\mathbf d^{*}=0$, makes $\Pi^{(d)}$ real, symmetric, and positive semidefinite, with $\mathrm{Tr}\,\Pi^{(d)}=1$.
	Equations~\eqref{eq:strongSOCid} and~\eqref{eq:strong-triplet} unify as
	\begin{equation}
			\boxed{\;\frac{\chi_{\mu\nu}(0)}{\chi_{N}}
			=\delta_{\mu\nu}-\Pi_{\mu\nu}+o(1),\;}
			\label{eq:unified}
		\end{equation}
		with $\Pi=\langle\hat{\mathbf g}\hat{\mathbf g}\rangle_{\rm FS}$ in the zero-field strong-SOC limit and $\Pi=\Pi^{(d)}$ in the exact zero-SOC unitary-triplet limit. The $o(1)$ term is identically zero in the latter case; in the SOC case it vanishes under the two scale limits in Eq.~\eqref{eq:strong-locking-scales}. No corresponding general substitution by the finite-field vector $\mu_B\mathbf H+\mathbf g_{\mathbf k}$ is valid.

	Equation~\eqref{eq:unified} is the controlled response identity underlying the geometric diagnostic. The next section places parity-mixed singlet-triplet states in the same zero-field language, after which Sec.~\ref{sec:intermediate} treats an arbitrary SOC-to-gap ratio for an $s$-wave gap within the reference-Fermi-surface regime. Section~\ref{sec:finiteH} explains why finite Zeeman field is not obtained by a simple replacement of the locking vector.

	\section[Parity-Mixed Pairing]{Parity-mixed pairing}
	\label{sec:paritymix}
	
	\subsection{Helicity-basis structure}
	\label{sec:paritymix-structure}
	
		In a noncentrosymmetric system, parity is not a good quantum number, and a generic superconducting state can mix even spin-singlet and odd spin-triplet components~\cite{Frigeri2004,Samokhin2007,Edelstein2008,bauer2012non,Smidman17}. The helicity-basis decomposition below requires only $\mathbf H=0$; the response identities in the following subsection additionally impose the two-sheet strong-locking assumptions of Eq.~\eqref{eq:strong-locking-scales}. Generalized unitarity is not imposed in this section, so the two helicity gaps may have unequal magnitudes. We use the pairing matrix in Eq.~\eqref{eq:Delta-gen}, with $\psi(-\mathbf{k})=\psi(\mathbf{k})$ and $\mathbf d(-\mathbf{k})=-\mathbf d(\mathbf{k})$. At $\mathbf{H}=0$, the natural spin-quantization axis is the helicity direction $\hat{\mathbf g}_{\mathbf{k}}=\mathbf{g}_{\mathbf{k}}/|\mathbf{g}_{\mathbf{k}}|$. Decompose the triplet part into longitudinal and transverse pieces,
	\begin{equation}
		\mathbf{d}(\mathbf{k})=d_{\parallel}(\mathbf{k})\hat{\mathbf g}_{\mathbf{k}}+\mathbf{d}_{\perp}(\mathbf{k}),\quad d_{\parallel}(\mathbf{k})=\mathbf{d}(\mathbf{k})\cdot\hat{\mathbf g}_{\mathbf{k}}.
		\label{eq:paritymix-dsplit}
	\end{equation}
	In the helicity basis of the normal-state Hamiltonian, the pairing matrix then takes the schematic form
	\begin{equation}
		\tilde\Delta(\mathbf{k})=
		\begin{pmatrix}
			\Delta_{+}(\mathbf{k}) & \Delta_{\perp}(\mathbf{k})\\
			-\Delta_{\perp}(\mathbf{k}) & \Delta_{-}(\mathbf{k})
		\end{pmatrix},
		\quad
		\Delta_{\pm}(\mathbf{k})=\psi(\mathbf{k})\pm d_{\parallel}(\mathbf{k}),
		\label{eq:paritymix-helicity-gap}
	\end{equation}
		where $\Delta_{\perp}(\mathbf{k})$ is set by the transverse component $\mathbf{d}_{\perp}(\mathbf{k})$ and mixes the two helicity sheets. The helicity-diagonal limit is exact when $\mathbf d_{\perp}=0$, including $\mathbf d\parallel\mathbf g$. If $\mathbf d_{\perp}\neq0$ but its magnitude is small compared with the helicity splitting $2|\mathbf g_{\mathbf k}|$, inter-helicity-pairing effects are perturbative and the helicity-diagonal response is the controlled leading approximation. In that limit, parity mixing changes the gap amplitudes on the $\pm$ sheets but not the underlying locking texture $\hat{\mathbf n}_{\mathbf{k}}=\hat{\mathbf g}_{\mathbf{k}}$.
	
	\subsection{$T=0$ consequences and finite-$T$ recovery}
	\label{sec:paritymix-consequences}
	
	We now compute $\chi_{\mu\nu}(T)$ in the helicity-diagonal limit $\Delta_{\perp}=0$. Rotating to the helicity basis with unitary $R_{\mathbf{k}}$ diagonalizing $H_{0}(\mathbf{k})$ at $\mathbf{H}=0$, the spin operator becomes
	\begin{equation}
		R^{\dagger}_{\mathbf{k}}(\hat e_{\mu}\cdot\hat\sigma)R_{\mathbf{k}}=(\hat{\mathbf g}_{\mathbf{k}}\cdot\hat e_{\mu})\tau_{3}+\sqrt{1-(\hat{\mathbf g}_{\mathbf{k}}\cdot\hat e_{\mu})^{2}}\,\tau_{\perp,\mu}(\mathbf{k}),
		\label{eq:paritymix-sigma-helicity}
	\end{equation}
		with $\tau_{3}$ diagonal and $\tau_{\perp,\mu}$ the normalized off-diagonal Hermitian helicity matrix for spin component $\mu$. In the helicity-diagonal limit, the BdG Hamiltonian decouples into one $2\times 2$ block for each helicity sheet, with gap $|\Delta_{\lambda}(\mathbf{k})|$ and standard BCS coherence factors $u_{\lambda},v_{\lambda}$. Substituting Eq.~\eqref{eq:paritymix-sigma-helicity} into the Kubo expression for $\chi_{\mu\nu}$ shows that the intra-helicity sector, driven by $\tau_{3}$, contributes a BCS-Yosida response on each sheet. The inter-helicity sector, driven by $\tau_{\perp,\mu}$, reproduces the strong-SOC locking-tensor calculation of Sec.~\ref{sec:anisoT0}. Adding the two gives
	\begin{equation}
		\begin{split}
			\frac{\chi_{\mu\nu}(T)}{\chi_{N}}
			={}& \delta_{\mu\nu}-\Pi_{\mu\nu} \\
				& + \sum_{\lambda=\pm}w_{\lambda}
				\left\langle
				Y_{\lambda}(\mathbf{k},T)\hat g_{\mu}(\mathbf{k})\hat g_{\nu}(\mathbf{k})
				\right\rangle_{\lambda} \\
					& + o(1),
		\end{split}
		\label{eq:paritymix-chiT}
	\end{equation}
			Here $\langle\cdots\rangle_{\lambda}$ is the FS average over the $\lambda$-helicity sheet, and $w_{\lambda}\equiv N_{\lambda}(0)/\sum_{\lambda'=\pm}N_{\lambda'}(0)$. The standard helicity-resolved Yosida function and corresponding quasiparticle energy are
	\begin{equation}
		\begin{aligned}
			Y_{\lambda}(\mathbf{k},T)&=-\!\int_{-\infty}^{\infty}\!d\xi\,\frac{\partial f}{\partial E}\bigg|_{E=E_{\lambda}(\xi,\mathbf{k})},\\
			E_{\lambda}(\xi,\mathbf{k})&=\sqrt{\xi^{2}+|\Delta_{\lambda}(\mathbf{k})|^{2}}.
		\end{aligned}
	\end{equation}
			respectively. The tensor $\Pi_{\mu\nu}\equiv\sum_{\lambda}w_{\lambda}\langle\hat g_{\mu}\hat g_{\nu}\rangle_{\lambda}$ is the helicity-resolved form of Eq.~\eqref{eq:Pidef} in the reference-Fermi-surface limit. The $o(1)$ term vanishes as $\Delta_{\max}/g_{\min}\to0$ and $g_{\max}/E_F\to0$; its detailed subleading form depends on the gaps and Fermi-surface weights. Sizable $\Delta_\perp$ requires a separate inter-helicity calculation.
	
	For a fully gapped state, $Y_{\lambda}\to 0$ as $T\to 0$, so
	\begin{equation}
			\chi_{\mu\nu}(0)/\chi_N
			=\delta_{\mu\nu}-\Pi_{\mu\nu}+o(1).
		\label{eq:paritymix-T0}
	\end{equation}
	A fully gapped helicity-diagonal parity mixture therefore does \emph{not} generate a new $T=0$ simplex point: it maps to the same locking-tensor geometry as the corresponding pure locking texture. What changes is the finite-$T$ recovery, because the two helicity sheets generally carry different gaps $\Delta_{\pm}$. Taking the trace of Eq.~\eqref{eq:paritymix-chiT},
	\begin{equation}
		\frac{\mathrm{Tr}\,\chi(T)}{\chi_{N}}
		=
		2+\sum_{\lambda=\pm}w_{\lambda}\langle Y_{\lambda}(\mathbf{k},T)\rangle_{\lambda}+o(1).
		\label{eq:paritymix-trace}
	\end{equation}
		This trace relation shows that distinct $\Delta_{\pm}$ on helicity sheets with appreciable weights $w_{\pm}$ can produce a broadened or resolved two-scale recovery while leaving the $T=0$ locking geometry unchanged. Translating this recovery to the measured Knight shift requires the controlled hyperfine mapping discussed later.
	
		Such recovery is consistent with a helicity-diagonal parity mixture but is neither necessary nor sufficient: it can also arise from multiband gaps, strong gap anisotropy or nodes, and pair-breaking impurity broadening. Conversely, it may be unresolved when $\Delta_{+}$ and $\Delta_{-}$ are similar, one helicity sheet has small weight, or experimental resolution is insufficient. Independent band-structure constraints, gap-anisotropy probes, and sample-quality assessment are therefore required. The finite-$T$ ellipsoid is a consistency check on parity admixture, not a standalone identification.
	
		If $\Delta_{\perp}$ is sizable, if one helicity sheet is unpaired or gapless, or if inter-helicity pairing dominates, Eq.~\eqref{eq:paritymix-T0} should be treated as a baseline rather than a controlled identity. A small but nonzero gap does not invalidate the strict $T=0$ result, although it may prevent measurements from reaching the asymptotic regime. In a finite Zeeman field, the paired states at $\mathbf k$ and $-\mathbf k$ generally have different spin axes, so the zero-field helicity reduction does not survive by replacing $\Pi$ with a field-dependent normal-state texture; the self-consistent BdG problem is required (Sec.~\ref{sec:finiteH}).

	\section[Intermediate SOC and Finite Field]{Beyond zero-field strong locking: intermediate SOC and finite field}
	\label{sec:beyond-locking}

	\subsection{Intermediate SOC at zero field}
	\label{sec:intermediate}
	
		The preceding results allow general FS shapes but assume the strict strong-locking limit. We now relax that limit at $\mathbf H=0$ for the analytically tractable pure $s$-wave state. Specifically, we allow $|\mathbf g_{\mathbf k}|/\Delta$ to be arbitrary while retaining the reference-Fermi-surface regime $g_{\max}\ll E_F$. Throughout this section, $\psi(\mathbf k)=\Delta$ and $\mathbf d(\mathbf k)=0$, with $\Delta$ denoting the singlet amplitude.
	
	Appendix~\ref{app:intermediate-support} gives the integral reduction and closed-form evaluation. Choosing the global gauge $\Delta>0$ and defining $\lambda_{\mathbf k}=|\mathbf g_{\mathbf k}|/\Delta$, the exact result is
	\begin{equation}
			\boxed{\;F_{s}(\lambda)\equiv1-\frac{\operatorname{asinh}\lambda}{\lambda\sqrt{1+\lambda^{2}}},\;}
		\label{eq:Fs}
	\end{equation}
	and
	\begin{equation}
		\boxed{\;\frac{\chi_{\mu\nu}(0)}{\chi_{N}}=\langle F_{s}(\lambda_{\mathbf{k}})[\delta_{\mu\nu}-\hat n_{\mu}\hat n_{\nu}]\rangle_{\rm FS}.\;}
		\label{eq:chi-explicit}
	\end{equation}

	Direct differentiation gives $F_s'(\lambda)>0$ for $\lambda>0$, with $F_s(0)=0$ and $\lim_{\lambda\to\infty}F_s(\lambda)=1$. Equation~\eqref{eq:chi-explicit} therefore retains the transverse locking tensor at each FS point but weights it by $F_s(\lambda_{\mathbf k})\in[0,1]$. Taking the trace gives
	\[
		\mathrm{Tr}\,\chi(0)/\chi_N
		=2\langle F_s(\lambda_{\mathbf k})\rangle_{\rm FS}\le 2,
	\]
	with strict inequality for any finite SOC scale. If the SOC magnitude is uniform over the FS, the strict monotonicity of $F_s$ permits the inversion
	\[
		\lambda=F_s^{-1}\!\left[\mathrm{Tr}\,\chi(0)/(2\chi_N)\right].
	\]
	For nonuniform SOC magnitude, the inversion instead defines the nonlinear FS average
	\begin{equation}
		\lambda_{\rm eff}\equiv F_{s}^{-1}\big(\langle F_{s}(\lambda_{\mathbf{k}})\rangle_{\rm FS}\big),
		\label{eq:lambda-eff}
	\end{equation}
	which need not equal $\lambda_{\mathbf k}$ at any FS point. Appendix~\ref{app:intermediate-support} gives its narrow-distribution expansion and Jensen bounds; Protocol~B gives the corresponding experimental inversion.
	
		Geometrically, define
		\[
			\Pi^{\rm eff}_{\mu\nu}
			=\langle F_s(\lambda_{\mathbf k})
			\hat n_\mu\hat n_\nu\rangle_{\rm FS},
		\]
		and let $p_i$ be its eigenvalues. Since $\mathrm{Tr}\,\Pi^{\rm eff}=S\equiv\langle F_s\rangle_{\rm FS}\le1$, they satisfy
		\[
			p_i\ge0,\qquad \sum_i p_i=S\le1.
		\]
			The intermediate-SOC map is therefore a three-dimensional tetrahedral region whose $S=1$ face is the strong-locking simplex and whose $S=0$ vertex is the origin [Fig.~\ref{fig:simplex}(c)]. At fixed $S$, the coordinates lie on a two-dimensional scaled simplex; for $S>0$, $S$ is the radial coordinate and $p_i/S$ gives the normalized simplex direction. For uniform SOC magnitude,
		\[
			\Pi^{\rm eff}=F_s(\lambda)\Pi,
		\]
		so $S=F_s(\lambda)$ and decreasing $\lambda$ moves a strong-locking point radially toward the origin without changing its normalized direction or principal axes. For nonuniform SOC,
		\[
			\Pi^{\rm eff}_{\mu\nu}
			=\langle F_s(\lambda_{\mathbf k})
		\hat n_\mu\hat n_\nu\rangle_{\rm FS},
	\]
	as defined in Eq.~\eqref{eq:Pi-eff}; correlations between SOC magnitude and texture can change both the eigenvalues and principal axes, so the contraction need not be radial.

		Appendix~\ref{app:intermediate-support} also provides a numerical cross-check against the published BdG curves of Ref.~\cite{pang2026} and a separately identified heuristic extension to anisotropic-gap unitary triplets. The digitized $T=0$ endpoints agree with Eq.~\eqref{eq:Fs} within $1.5\times10^{-3}$ in $\chi/\chi_N$. We next turn from this zero-field interpolation to the qualitatively different finite-field problem.

	\subsection{Finite Zeeman field: limitations and numerical regime}
	\label{sec:finiteH}
	
	Section~\ref{sec:intermediate} relies on pairing between time-reversed zero-field helicity partners. A Zeeman field changes this structure fundamentally. Recalling Eq.~\eqref{eq:ndef}, and denoting the finite-field normal-state unit texture by $\hat{\mathbf n}^{(b)}_{\mathbf k}$, we write
	\begin{equation}
		\mathbf b_{\mathbf k}=\mu_B\mathbf H+\mathbf g_{\mathbf k},
		\qquad
		\hat{\mathbf n}^{(b)}_{\mathbf k}
		=\mathbf b_{\mathbf k}/|\mathbf b_{\mathbf k}|.
		\label{eq:bvec}
	\end{equation}
	Because $\mathbf g_{-\mathbf k}=-\mathbf g_{\mathbf k}$,
	\begin{equation}
		\mathbf b_{-\mathbf k}=\mu_B\mathbf H-\mathbf g_{\mathbf k}
		\ne-\mathbf b_{\mathbf k}
		\quad(\mathbf H\ne0),
		\label{eq:asymmetry}
	\end{equation}
	in general. Thus, at nonzero field, the spin axes at $\mathbf k$ and $-\mathbf k$ are not generically antiparallel, while
	\[
		|\mathbf b_{\mathbf k}|^2-|\mathbf b_{-\mathbf k}|^2
		=4\mu_B\mathbf H\cdot\mathbf g_{\mathbf k}
	\]
	shows that their normal-state splittings are not generically equal. Consequently, the condition $|\mathbf b_{\mathbf k}|\gg\Delta$ alone does not yield the two independent $2\times2$ helicity blocks used at zero field, and the substitution $\hat{\mathbf g}_{\mathbf k}\mapsto\hat{\mathbf n}^{(b)}_{\mathbf k}$ in Eq.~\eqref{eq:unified} is unjustified.
	
	The finite-field texture tensor remains well defined as the normal-state descriptor
	\begin{equation}
		\Pi^{(b)}_{\mu\nu}(\mathbf H)
		\equiv
		\left\langle
		\hat n^{(b)}_{\mu}(\mathbf k;\mathbf H)
		\hat n^{(b)}_{\nu}(\mathbf k;\mathbf H)
		\right\rangle_{\rm FS},
		\label{eq:Pi-H}
	\end{equation}
		which equals $\Pi_{\mu\nu}(\mathbf H,\mathbf g)$ of Eq.~\eqref{eq:Pidef}. It is only a normal-state geometric descriptor and does not determine the superconducting susceptibility. Appendix~\ref{app:finite-field-response} derives its weak-field expansion and discusses the model-specific cancellation and Bogoliubov-Fermi-surface scales.
	
		Within mean-field BdG theory at a generic field, the order parameter must be determined self-consistently and the full $4\times4$ BdG Hamiltonian diagonalized. A Kubo calculation about a finite-field background gives a differential response, whereas the Knight shift measures the equilibrium secant response proportional to $\hat{\mathbf H}\cdot\mathbf M(T,\mathbf H)/H$; they coincide only in linear response. Appendix~\ref{app:finite-field-response} gives the precise definitions and order-parameter-relaxation term, and Ref.~\cite{pang2026} provides explicit finite-field calculations. Protocol~C therefore compares measured field sweeps with self-consistent calculations of the same finite-field Knight-shift observable rather than with $\Pi^{(b)}(\mathbf H)$.
	
	\section[Multiband Extension]{Multiband extension: decoupled-pocket baseline}
	\label{sec:multiband}
		
		We next extend the analysis to band and orbital structure. We first derive the additive disjoint-pocket form, then state its validity domain and apply it to a conditional A$_2$Cr$_3$As$_3$ baseline.
	
	\subsection{Disjoint-pocket form}
	
	Let $i$ label disjoint Fermi-surface pockets and define $\chi_N=\sum_i\chi_N^{(i)}$. If each pocket independently satisfies the zero-field, fully gapped, helicity-diagonal strong-SOC assumptions, additivity gives
	\begin{equation}
		\chi_{\mu\nu}(0)
		=\sum_i\chi_N^{(i)}
		[\delta_{\mu\nu}-\Pi_{\mu\nu}^{(i)}]+\chi_No(1),
		\label{eq:multiband-id}
	\end{equation}
	where the aggregate remainder is understood componentwise. In the exact zero-SOC unitary-triplet limit, the same form holds with $\Pi^{(i)}=\Pi^{(d,i)}$, and the remainder vanishes identically. Hence
	\begin{equation}
		\mathrm{Tr}\,\chi(0)
		=2\chi_N+\chi_No(1).
		\label{eq:multiband-trace}
	\end{equation}
	This is a controlled decoupled-pocket identity, not a general multiband inequality.
	
	\subsection{Scope and failure modes of the pocket baseline}
		
	Equation~\eqref{eq:multiband-id} applies only when pairing and the spin-response vertices are block diagonal in pocket space and each pocket satisfies a controlled response identity. Interband pairing, orbital hybridization, or interband spin-vertex matrix elements generally generate cross terms not represented by $\Pi^{(i)}$. The exact Bogoliubov sum rule remains valid in the full spin-orbital Hilbert space, but by itself implies neither Eq.~\eqref{eq:multiband-trace} nor a pocketwise simplex decomposition. A deviation therefore rejects the joint baseline rather than uniquely diagnosing interband pairing; interpretation requires a dedicated multiorbital response calculation.
	
	\subsection{Conditional baseline for A$_{2}$Cr$_{3}$As$_{3}$}
	\label{sec:multiband-A2Cr3As3}
	
	The A$_{2}$Cr$_{3}$As$_{3}$ family ($A=$ Na, K, Rb, Cs) comprises quasi-one-dimensional superconductors~\cite{Bao15,Tang15_1,Tang15_2,Mu18,ZHOU2017208}. Density-functional calculations report three Fermi-surface sheets, conventionally labeled $\alpha,\beta,\gamma$, with quasi-1D $\alpha,\beta$ sheets and a more three-dimensional $\gamma$ sheet~\cite{Jiang15,Wu15}. Within the additive pocket baseline, we define the scalar normal-state susceptibility weights $w_i=\chi_N^{(i)}/\sum_j\chi_N^{(j)}$, so $\sum_iw_i=1$. These weights need not equal bare density-of-states fractions when pocket-dependent magnetic renormalizations or spin-vertex factors are present.
	We now define a phenomenological two-texture baseline, not a consequence of crystal symmetry or of the cited DFT calculations. Assume
	\[
		\Pi^{(\alpha)}\simeq\Pi^{(\beta)}
		\simeq\operatorname{diag}(0,0,1),
		\qquad
		\Pi^{(\gamma)}\simeq\mathbb I/3,
	\]
	with every pocket in the zero-field, fully gapped, helicity-diagonal strong-SOC regime of Eq.~\eqref{eq:multiband-id}. With $w_\parallel=w_\alpha+w_\beta=1-w_\gamma$, the susceptibility-weighted locking tensor is
	\[
		\Pi_{\rm eff}\equiv\sum_iw_i\Pi^{(i)}
		\simeq w_\parallel\operatorname{diag}(0,0,1)
		+\frac{w_\gamma}{3}\mathbb I,
	\]
	with $\operatorname{Tr}\Pi_{\rm eff}=1$, so
	\begin{equation}
		\frac{\chi_{\mu\nu}(0)}{\chi_N}
		=\delta_{\mu\nu}-(\Pi_{\rm eff})_{\mu\nu}+o(1).
		\label{eq:a2cr3as3-pred}
	\end{equation}
		Componentwise,
	\begin{align}
		\frac{\chi_{xx}(0)}{\chi_N}
		=\frac{\chi_{yy}(0)}{\chi_N}
		&=\frac23+\frac13w_\parallel+o(1),
		\label{eq:a2cr3as3-perp}\\
		\frac{\chi_{zz}(0)}{\chi_N}
		&=\frac23(1-w_\parallel)+o(1).
		\label{eq:a2cr3as3-par}
	\end{align}
	As $w_\parallel$ varies from $0$ to $1$, this conditional baseline traces the simplex line from the barycenter to the $\hat c$ vertex. A discrepancy rejects this specified joint baseline; it does not by itself rule out a general band-resolved SOC mechanism or uniquely establish interband pairing.

	\section[Core Diagnostic]{Core diagnostic: Knight-shift ellipsoid}
	\label{sec:ellipsoid}
	
		In either the strong-SOC regime of Eq.~\eqref{eq:strong-locking-scales} or the exact zero-SOC unitary-triplet limit, Eq.~\eqref{eq:unified} relates the electronic susceptibility to the locking tensor $\Pi$. Only after orbital subtraction, controlled hyperfine mapping, and verification of the controlled-locking trace identity can the measured spin Knight-shift tensor be mapped to $\Pi$.

	Figure~\ref{fig:simplex} previews the geometry developed in this section. The $\kappa_i$ are the normalized principal Knight-shift semi-axes. In the controlled-locking limits, $\kappa_i=1-\pi_i+o(1)$, where $\pi_i$ are the eigenvalues of $\Pi$. For intermediate-SOC $s$-wave pairing, $p_i$ are the eigenvalues of $\Pi^{\rm eff}$ and $S=\mathrm{Tr}\,\Pi^{\rm eff}=\tfrac12\sum_i\kappa_i$, so $\kappa_i=S-p_i$. The definitions, validity conditions, and hyperfine mapping are given below.
		
	\begin{figure*}[tb]
		\centering
			\begin{minipage}[c]{0.78\textwidth}
		\centering
		\begin{tikzpicture}[x=1.3cm,y=1.3cm,font=\small]
			
			\tikzset{
				refsphere/.style = {gray!55, dashed, line width=0.4pt},
				ellfill/.style   = {fill=blue!15, fill opacity=0.55},
				ellline/.style   = {draw=blue!55!black, line width=0.7pt},
				ellback/.style   = {draw=blue!55!black, line width=0.6pt, dashed},
				axisarr/.style   = {->, gray!75!black, line width=0.45pt, >=stealth},
				axislab/.style   = {font=\footnotesize, gray!45!black, inner sep=1pt},
				panellab/.style  = {font=\small},
				panelco/.style   = {font=\footnotesize}
			}
			
			\def\fore{0.32}
			\def\ax{0.45}
			
			\newcommand{\refSphere}{%
				\draw[refsphere] (0,0) circle (1);
				\draw[refsphere] (-1,0) arc (180:360:1 and \fore);
				\draw[refsphere] (-1,0) arc (180:0:1 and \fore);
			}
			
			\newcommand{\drawAxes}{%
				\draw[axisarr] (0,0) -- (\ax,0)
				node[axislab, right=0pt] {$\hat a$};
				\draw[axisarr] (0,0) -- (-\ax*0.65,-\ax*0.5)
				node[axislab, below left=-2pt] {$\hat b$};
				\draw[axisarr] (0,0) -- (0,\ax)
				node[axislab, above=0pt] {$\hat c$};
			}
			
			\newcommand{\drawEllipsoid}[3]{%
				\path[ellfill] (0,0) ellipse ({#1} and {#3});
				\draw[ellline] (0,0) ellipse ({#1} and {#3});
				\draw[ellback]  ({-#1},0) arc (180:360:{#1} and {\fore*#2});
				\draw[ellline]  ({-#1},0) arc (180:0:{#1} and {\fore*#2});
			}
			
			\newcommand{\drawDisc}[2]{%
				\path[ellfill] (0,0) ellipse ({#1} and {\fore*#2});
				\draw[ellline] (0,0) ellipse ({#1} and {\fore*#2});
			}
			
			\def\sep{2.6}
			
			\begin{scope}[shift={(0,0)}]
				\refSphere
				\drawEllipsoid{0.667}{0.667}{0.667}
				\drawAxes
				\node[panellab] at (0,-1.55) {(i) Sphere};
				\node[panelco] at (0,-1.95) {$(\tfrac{1}{3},\tfrac{1}{3},\tfrac{1}{3})$};
			\end{scope}
			
			\begin{scope}[shift={(\sep,0)}]
				\refSphere
				\drawDisc{1.0}{1.0}
				\drawAxes
				\node[panellab] at (0,-1.55) {(ii) Oblate};
				\node[panelco] at (0,-1.95) {$(0,0,1)$};
			\end{scope}
			
			\begin{scope}[shift={(2*\sep,0)}]
				\refSphere
				\drawEllipsoid{0.5}{0.5}{1.0}
				\drawAxes
				\node[panellab] at (0,-1.55) {(iii) Prolate};
				\node[panelco] at (0,-1.95) {$(\tfrac{1}{2},\tfrac{1}{2},0)$};
			\end{scope}
			
			\begin{scope}[shift={(3*\sep,0)}]
				\refSphere
				\drawEllipsoid{0.5}{0.7}{0.8}
				\drawAxes
				\node[panellab] at (0,-1.55) {(iv) Triaxial};
				\node[panelco] at (0,-1.95) {$(\tfrac{1}{2},\tfrac{3}{10},\tfrac{1}{5})$};
			\end{scope}
			
		\end{tikzpicture}
			\par\smallskip\textbf{(a) Canonical ellipsoids}
			\end{minipage}
			\par\medskip
			\begin{minipage}[t]{0.45\textwidth}
			\centering
			\begin{tikzpicture}[x=0.82cm,y=0.82cm,font=\small]
			\coordinate (V1) at (-2.8, 0);
			\coordinate (V2) at ( 2.8, 0);
			\coordinate (V3) at ( 0, 4.8497422);
			\coordinate (M12) at ( 0, 0);
			\coordinate (M13) at (-1.4, 2.4248711);
			\coordinate (M23) at ( 1.4, 2.4248711);
			\coordinate (B) at ( 0, 1.6165808);
			\coordinate (G) at (-0.56, 0.9699484);
			\fill[gray!8] (V1) -- (V2) -- (V3) -- cycle;
			\draw[line width=1pt] (V1) -- (V2) -- (V3) -- cycle;
			\draw[gray!55, dashed, thin] (V1) -- (M23);
			\draw[gray!55, dashed, thin] (V2) -- (M13);
			\draw[gray!55, dashed, thin] (V3) -- (M12);
			\fill (V1) circle (0.055);
			\fill (V2) circle (0.055);
			\fill (V3) circle (0.055);
			\fill[red!80!black] (M12) circle (0.055);
			\fill[red!80!black] (M13) circle (0.055);
			\fill[red!80!black] (M23) circle (0.055);
			\fill[blue!70!black] (B) circle (0.075);
			\fill[orange!80!black] ($(G)+(-0.065,-0.065)$) rectangle ($(G)+(0.065,0.065)$);
			\draw[green!30!black, line width=0.3pt]
			($(V3)+(0,0.13)$) --
			($(V3)+(0.13,0)$) --
			($(V3)+(0,-0.13)$) --
			($(V3)+(-0.13,0)$) -- cycle;
			\node[below left=2pt, font=\small] at (V1) {$\pi_{1}=1$};
			\node[below right=2pt, font=\small] at (V2) {$\pi_{2}=1$};
			\node[above=3pt, font=\small] at (V3) {$\pi_{3}=1$};
			\node[above=13pt, font=\small, green!35!black] at (V3) {K$_{2}$Cr$_{3}$As$_{3}$ (cond.)};
			\node[right=3pt, font=\small, align=left, inner sep=1pt]
			at (B) {(i)\\barycenter};
			\node[below=3pt, font=\small] at (M12) {(iii) edge midpoint};
			\node (iilab) [font=\small] at (-1.15, 4.30) {(ii) vertex};
			\draw[gray!60!black, line width=0.3pt, shorten <=1.5pt, shorten >=2.5pt]
			(iilab.east) -- (V3);
			\node (ivlab) [font=\small, orange!70!black] at (-0.56, 0.44) {(iv)};
			\draw[orange!60!black, line width=0.3pt, shorten <=1.5pt, shorten >=2.5pt]
			(ivlab.north) -- (G);
			\end{tikzpicture}
				\par\smallskip\textbf{(b) Simplex map}
				\end{minipage}
				\hfill
				\begin{minipage}[t]{0.52\textwidth}
			\centering
			\begin{tikzpicture}[font=\small]
				\coordinate (O) at (0,-1.15);
				\coordinate (P1) at (-2.2,0);
				\coordinate (P2) at (2.2,0);
				\coordinate (P3) at (0,2.8);
				\coordinate (B) at (0,0.933333);
				\coordinate (C) at (-0.85,0.82);

				\fill[gray!9] (P1) -- (P2) -- (P3) -- cycle;
				\draw[line width=0.9pt] (P1) -- (P2) -- (P3) -- cycle;
				\draw[gray!65, line width=0.7pt] (O) -- (P1);
				\draw[gray!65, line width=0.7pt] (O) -- (P2);
				\draw[gray!65, line width=0.7pt] (O) -- (P3);

				\coordinate (A1) at ($(O)!0.42!(P1)$);
				\coordinate (A2) at ($(O)!0.42!(P2)$);
				\coordinate (A3) at ($(O)!0.42!(P3)$);
				\coordinate (D1) at ($(O)!0.70!(P1)$);
				\coordinate (D2) at ($(O)!0.70!(P2)$);
				\coordinate (D3) at ($(O)!0.70!(P3)$);
				\fill[blue!10, opacity=0.55] (A1) -- (A2) -- (A3) -- cycle;
				\draw[blue!55!black, dashed, line width=0.6pt] (A1) -- (A2) -- (A3) -- cycle;
				\fill[blue!7, opacity=0.45] (D1) -- (D2) -- (D3) -- cycle;
				\draw[blue!55!black, dashed, line width=0.6pt] (D1) -- (D2) -- (D3) -- cycle;

				\draw[->, blue!70!black, line width=1.0pt] (B) -- (O);
				\draw[->, orange!80!black, densely dashed, line width=0.9pt]
				(C) .. controls (-1.05,0.25) and (-0.55,-0.35) .. (O);

				\fill (P1) circle (0.045);
				\fill (P2) circle (0.045);
				\fill (P3) circle (0.045);
				\fill (O) circle (0.045);
				\node[left=2pt] at (P1) {$p_1=1$};
				\node[right=2pt] at (P2) {$p_2=1$};
				\node[above=2pt] at (P3) {$p_3=1$};
				\node[below=2pt] at (O) {$p=(0,0,0)$, $S=0$};
				\node (nonuniformlabel) [orange!80!black, align=right] at (-2.45,1.18)
				{nonuniform\\SOC};
				\draw[orange!80!black, line width=0.4pt]
				(nonuniformlabel.east) -- (C);
				\node (stronglabel) at (1.55,1.75) {$S=1$};
				\draw[gray!70, line width=0.4pt]
				(stronglabel.west) -- (B);
				\node (uniformlabel) [blue!70!black, align=left] at (2.65,0.90)
				{uniform SOC\\(radial)};
				\draw[blue!70!black, line width=0.4pt]
				(uniformlabel.west) -- ($(B)!0.35!(O)$);
				\node (s70label) [blue!60!black] at (2.45,-0.42) {$S=0.70$};
				\draw[blue!60!black, line width=0.4pt]
				(s70label.west) -- (D2);
				\node (s42label) [blue!60!black] at (2.10,-0.92) {$S=0.42$};
				\draw[blue!60!black, line width=0.4pt]
				(s42label.west) -- (A2);
			\end{tikzpicture}
			\par\smallskip\textbf{(c) Intermediate-SOC tetrahedron}
			\end{minipage}
				\caption{Knight-shift ellipsoid and locking-coordinate maps at $T=0$. (a) Canonical shapes in the controlled locking limits, with the faint dashed spheres denoting the normalized normal-state references and the solid blue forms denoting the superconducting ellipsoids with leading-order semi-axes $\kappa_i(0)=1-\pi_i$: (i) the isotropic sphere at the barycenter, (ii) the oblate disc at a vertex, (iii) the prolate ellipsoid at an edge midpoint, and (iv) a generic triaxial ellipsoid in the interior. (b) The two-dimensional strong-locking simplex $\sum_i\pi_i=1$, $\pi_i\in[0,1]$; the numbered markers connect representative simplex points to their shapes in panel (a). Vertices describe one-axis locking, edge midpoints describe planar locking, the barycenter describes isotropic three-dimensional locking, and generic interior points describe lower-symmetry textures. Pairing labels require additional assumptions. The open green diamond indicates the conditional fixed-tensor zero-field placement of K$_{2}$Cr$_{3}$As$_{3}$ inferred from finite-field data, not a measured simplex coordinate; the orange square marks $(\tfrac{1}{2},\tfrac{3}{10},\tfrac{1}{5})$. (c) Three-dimensional intermediate-SOC space, where $p_i$ are the eigenvalues of $\Pi^{\rm eff}$ and $S\equiv\mathrm{Tr}\,\Pi^{\rm eff}=\langle F_s(\lambda_{\mathbf k})\rangle_{\rm FS}=\tfrac12\sum_i\kappa_i$ is the SOC-weighted radial response coordinate [Eqs.~\eqref{eq:Fs} and~\eqref{eq:Pi-eff}]. Thus $p_i\ge0$ and $\sum_i p_i=S\le1$. The outer face $S=1$ is the strong-locking simplex, dashed triangles are fixed-$S$ scaled simplices, and the origin is the zero-SOC collapse. Uniform SOC gives the radial trajectory $p_i=S\pi_i$; nonuniform SOC can follow a nonradial path. The corresponding Knight-shift semi-axes are $\kappa_i=S-p_i$.}
		\label{fig:simplex}
	\end{figure*}
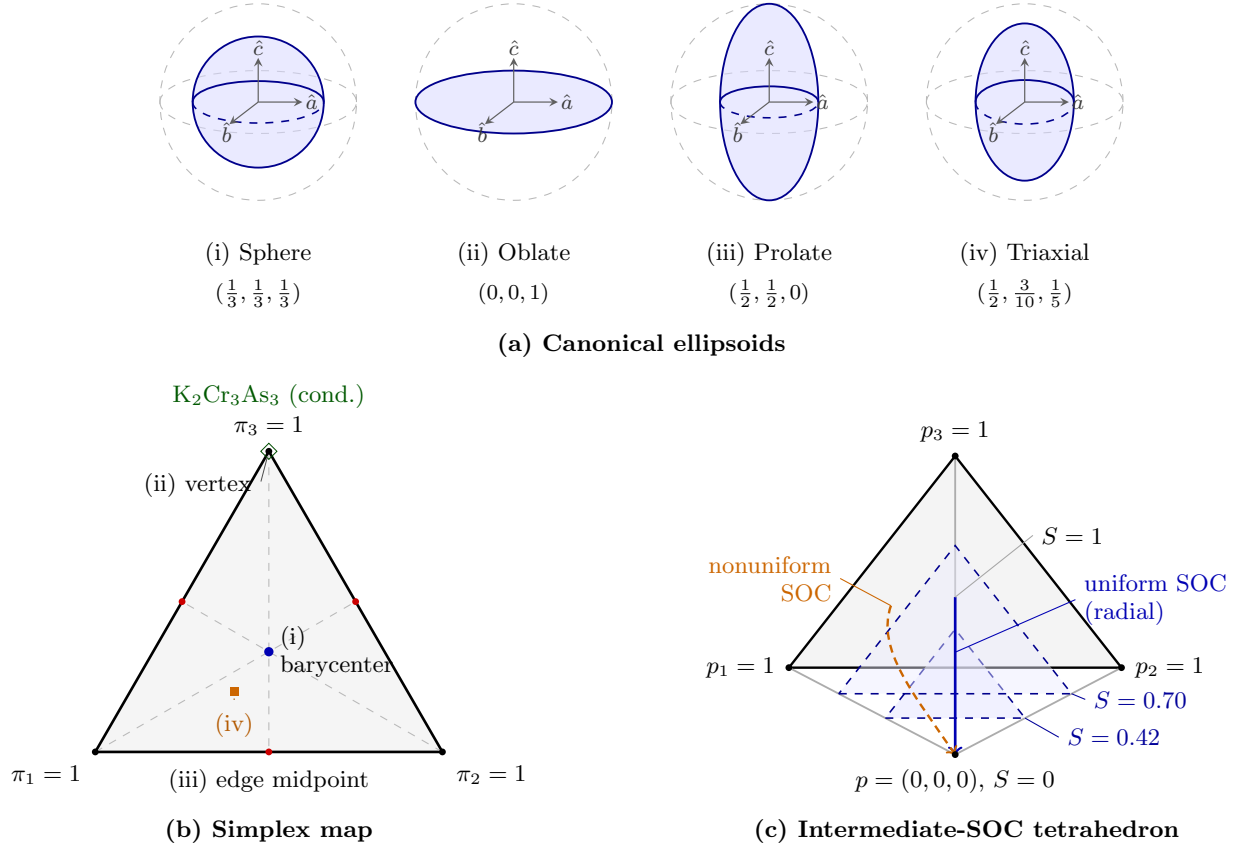

	\subsection{Reference identities and ellipsoid definition}
	\label{sec:keyresults}
	
	Table~\ref{tab:keyresults} collects the central results, validity domains, and experimental uses.
		
	\begin{table*}[tbp]
		\caption{Central identities, validity domains, and uses. Only the Bogoliubov sum rule is kinematic; response statements require the listed assumptions.}
		\label{tab:keyresults}
		\small
		\renewcommand\arraystretch{1.4}
		\setlength{\tabcolsep}{1.5ex}
		\centering
			\begin{tabular*}{\textwidth}{@{\extracolsep{\fill}} c l p{0.28\textwidth} p{0.20\textwidth} p{0.17\textwidth} @{}}
			\hline\hline
			\# & Result & \rcol{Statement} & \rcol{Domain of validity} & \rcol{Used in}\\
			\hline
			1 & Bogoliubov sum rule & \rcol{$\sum_{s_{1}s_{2}}[W^{s_{1}s_{2}}_{ph,O}+W^{s_{1}s_{2}}_{pp,O}]=\mathrm{Tr}_{s}(O^{2})$} & \rcol{Hermitian $O$; any $\mathbf{k}$} & \rcol{Eq.~\eqref{eq:general-sum-rule}; basis for later rows}\\
			2 & Pointwise pp budget & \rcol{$\sum_{s_1s_2}W^{s_1s_2}_{pp,O}(\mathbf k)\le\mathrm{Tr}_s(O^2)$} & \rcol{Hermitian $O$; any $\mathbf k$} & \rcol{Eq.~\eqref{eq:p-pointwise}; no response bound alone}\\
			3 & Controlled residual traces & \rcol{$\mathrm{Tr}\,\chi(0)/\chi_N=2+o(1)$ (locking; exact at zero SOC); $\le2$ (intermediate-SOC $s$ wave)} & \rcol{Zero field; full gap; stated controlled limits} & \rcol{Eqs.~\eqref{eq:unified} and~\eqref{eq:chi-explicit}; trace test}\\
			4 & Site-resolved budget & \rcol{$\sum(W_{ph,O_{{\rm hf},\alpha}}+W_{pp,O_{{\rm hf},\alpha}})=2\sum_{\mu}(A^{(\mathbf{R})}_{\alpha\mu})^{2}$} & \rcol{Real hyperfine tensor; any $\mathbf{k}$} & \rcol{Eq.~\eqref{eq:site-budget}; local operator budget}\\
			5 & Unified locking identity & \rcol{$\chi_{\mu\nu}(0)/\chi_N=\delta_{\mu\nu}-\Pi_{\mu\nu}+o(1)$} & \rcol{Full gap; strong-SOC helicity diagonal or exact zero-SOC unitary triplet} & \rcol{Eq.~\eqref{eq:unified}; Protocol B}\\
			6 & Trace from ellipsoid & \rcol{$\sum_i\kappa_i(0)=2+o(1)$} & \rcol{Unified locking regime; scalar or axis-normalized coaxial hyperfine} & \rcol{Eq.~\eqref{eq:trace-from-ellipsoid}; Protocol B}\\
			7 & Intermediate-SOC, $s$ wave & \rcol{$\chi_{\mu\nu}(0)/\chi_N=\langle F_s(\lambda_{\mathbf k})[\delta_{\mu\nu}-\hat n_\mu\hat n_\nu]\rangle_{\rm FS}$} & \rcol{$\mathbf H=0$; $s$ wave; arbitrary $|\mathbf g|/\Delta$; $g_{\max}\ll E_F$} & \rcol{Eq.~\eqref{eq:chi-explicit}; Protocol B}\\
			8 & $F_s$ closed form & \rcol{$F_s(\lambda)=1-\operatorname{asinh}\lambda/[\lambda\sqrt{1+\lambda^2}]$} & \rcol{As in row 7} & \rcol{Eq.~\eqref{eq:Fs}; SOC inversion; $\lambda_{\rm eff}$ if nonuniform}\\
			9 & Finite-$\mathbf H$ response & \rcol{Self-consistent BdG generally required; $\Pi^{(b)}(\mathbf H)$ is a normal-state descriptor} & \rcol{Finite $\mathbf H$; $\mathbf b_{-\mathbf k}\ne-\mathbf b_{\mathbf k}$; stable superconducting state} & \rcol{Eqs.~\eqref{eq:asymmetry} and~\eqref{eq:Pi-H-weak}; Protocol C}\\
			10 & Kramers--Kronig subtraction & \rcol{$\tfrac{2}{\pi}\lim_{\mathbf q\to0}\int_0^\infty\!\frac{d\omega}{\omega}[\mathrm{Im}\chi_{\mu\mu}^{R,N}-\mathrm{Im}\chi_{\mu\mu}^{R,SC}]=\chi_{\mu\mu}^{N}(0)-\chi_{\mu\mu}^{SC}(0)$} & \rcol{Causal response; static $\mathbf q\to0$ limit} & \rcol{Eq.~\eqref{eq:spinFGT}; signed, not conserved}\\
			11 & \parbox[t]{0.24\textwidth}{\raggedright Vanishing of the $\mathbf q=0$ contribution to $1/T_1$\par} & \rcol{$\Pi^{(g)}_{\beta\beta}=\Pi^{(g)}_{\gamma\gamma}=0\Rightarrow(1/T_1T)_{\mathbf q=0,H\parallel\alpha}^{N,SC}=0$} & \rcol{Strong-SOC helicity locking; weak field; calibrated transverse channels} & \rcol{Eq.~\eqref{eq:T1-vanishing}; Protocol D; conditional test}\\
			12 & Decoupled-pocket identity & \rcol{$\chi_{\mu\nu}(0)=\sum_i\chi_N^{(i)}[\delta_{\mu\nu}-\Pi_{\mu\nu}^{(i)}]+\chi_No(1)$; $\mathrm{Tr}\,\chi(0)=2\chi_N+\chi_No(1)$} & \rcol{Decoupled full-gap pockets; invalidated by relevant interband pairing, hybridization, or spin vertices} & \rcol{Eqs.~\eqref{eq:multiband-id} and~\eqref{eq:multiband-trace}; Protocol F}\\
			\hline\hline
			\end{tabular*}
	\end{table*}
	
	For scalar hyperfine coupling $A^{(\mathbf{R})}_{\mu\nu}=A^{(\mathbf{R})}\delta_{\mu\nu}$ (the default setting of this subsection), the spin part of the Knight-shift tensor at a fixed site $\mathbf{R}$ is related in linear response to $\chi_{\mu\nu}(T)\equiv\left.\partial M_\mu/\partial H_\nu\right|_{\mathbf H=0}$ by~\cite{Slichter1990,tinkham,Triplet2021}
	\begin{equation}
		K^{\rm spin}_{\mu\nu}(T)=A^{(\mathbf{R})}\chi_{\mu\nu}(T).
		\label{eq:Kdef}
	\end{equation}
	For the remainder of this subsection, we suppress the site label and write $A\equiv A^{(\mathbf R)}$. We define $K_N^{\rm spin}\equiv A\chi_N$ and
	\begin{equation}
		\kappa(T)\equiv
		\frac{K^{\rm spin}(T)}{K_N^{\rm spin}}
		=\frac{\chi(T)}{\chi_N}.
		\label{eq:kappa-def}
	\end{equation}
	The diagonal anisotropic case is treated in Sec.~\ref{sec:ellipsoid-hf}.
	
	The \emph{Knight-shift ellipsoid} is the image of the unit ball under the normalized linear-response tensor $\kappa(T)$. It is related to the susceptibility indicatrix of Nye~\cite{Nye_book}, but its semi-axes are the response eigenvalues themselves.
	
	In the linear-response regime, rotation of a single crystal relative to a field direction $\hat{\mathbf{H}}$ measures
	\begin{equation}
		K^{\rm spin}(\hat{\mathbf{H}},T)
		=\hat{\mathbf H}\cdot K^{\rm spin}(T)\cdot\hat{\mathbf H}.
		\label{eq:Kscalar}
	\end{equation}
	At fields where the response depends appreciably on field magnitude or orientation, rotation instead traces a finite-field Knight-shift surface governed by $\mathbf M(\mathbf H)/H$; it need not admit an ellipsoid representation (Appendix~\ref{app:finite-field-response} and Protocol~C).
	The symmetric tensor $\kappa(T)$ has principal values $\kappa_1,\kappa_2,\kappa_3\ge0$ with orthonormal principal axes. Define the solid ellipsoid
	\begin{equation}
		\mathcal E_K(T)
		\equiv
		\{\mathbf x=\kappa(T)\mathbf u:\|\mathbf u\|\le1\}.
		\label{eq:ellipsoid}
	\end{equation}
	Its semi-axes are $\kappa_i(T)$, so $\kappa_i=0$ collapses the corresponding axis. Six independent directional measurements at fixed $T$ in the linear-response regime determine a general symmetric $\kappa(T)$.
	
	\subsection{Controlled-locking simplex and allowed shapes}
	
		Substituting the unified locking identity Eq.~\eqref{eq:unified} into Eqs.~\eqref{eq:Kdef}--\eqref{eq:Kscalar} gives
	\begin{equation}
		\boxed{\;
		\frac{K^{\rm spin}(\hat{\mathbf H},0)}{K_N^{\rm spin}}
		=1-\hat{\mathbf H}\cdot\Pi\cdot\hat{\mathbf H}+o(1).
		\;}
		\label{eq:Kellipsoid}
	\end{equation}
	The principal-axis frame of $\mathcal E_K(0)$ coincides with that of $\Pi$, and
	\begin{equation}
		\kappa_i(0)=1-\pi_i+o(1),\qquad i=1,2,3.
		\label{eq:Ki-pi}
	\end{equation}
	The three eigenvalues $\pi_{i}\in[0,1]$ of $\Pi$ satisfy $\pi_{1}+\pi_{2}+\pi_{3}=\mathrm{Tr}\,\Pi=1$, so the $T=0$ principal values obey
	\begin{equation}
		\boxed{\;\sum_i\kappa_i(0)=2+o(1).\;}
		\label{eq:trace-from-ellipsoid}
	\end{equation}
	The $o(1)$ remainder vanishes identically in the exact zero-SOC unitary-triplet limit.
		This simplex classification is a statement about the locking tensor, not by itself about the parity label. As Sec.~\ref{sec:paritymix} shows, a fully gapped helicity-diagonal parity-mixed state maps to the same $T=0$ ellipsoid as the corresponding pure locking texture; distinguishing parity admixture then requires finite-$T$ information or other probes.
	
		The constraints $\pi_{1}+\pi_{2}+\pi_{3}=1$ and $0\le\pi_{i}\le1$ confine the locking-tensor eigenvalues to a two-dimensional simplex (Fig.~\ref{fig:simplex}). Each point fixes an eigenvalue pattern and hence an ellipsoid shape, but it does not uniquely determine a pairing class or microscopic locking texture. Figure~\ref{fig:simplex} shows the canonical geometries, while Table~\ref{tab:diagnose} collects their experimental readout and principal alternatives.
	
	Equation~\eqref{eq:trace-from-ellipsoid} applies only in the controlled locking limits. The exact intermediate-SOC $s$-wave result instead gives $\sum_i\kappa_i(0)=2\langle F_s\rangle_{\rm FS}\le2$, while other pairings can lie on either side of $2$; a deviation therefore diagnoses failure of the controlled locking description, not a unique pairing class. Likewise, an ellipsoid shape does not identify pairing parity. Different nuclei must return the same normalized simplex point only for scalar hyperfine coupling or for coaxial diagonal tensors normalized axis by axis; a general non-coaxial hyperfine tensor must be calibrated independently, and the exact site budget Eq.~\eqref{eq:site-budget} does not remove this requirement.

		\subsection{Exact cubic consequences and UBe$_{13}$ illustration}
		\label{sec:ellipsoid-cubic}

		Within the controlled locking limits, the simplex construction applies without change to cubic crystals, but cubic crystal symmetry alone does not imply a generic interior point.

		\begin{corollary}[Cubic barycenter]
		\label{cor:cubic-barycenter}
		If the superconducting state and its spin-locking texture preserve the full cubic point group, then
		\begin{equation}
			R\Pi R^{T}=\Pi
		\qquad\text{for every cubic rotation }R.
		\label{eq:cubic-invariance}
	\end{equation}
	The only symmetric second-rank tensor satisfying Eq.~\eqref{eq:cubic-invariance} is proportional to the identity. Since $\mathrm{Tr}\,\Pi=1$,
	\begin{equation}
		\boxed{\begin{aligned}
			\Pi&=\tfrac{1}{3}\mathbb I,\\
			(\pi_1,\pi_2,\pi_3)
			&=(\tfrac13,\tfrac13,\tfrac13),\\
			\frac{\chi(0)}{\chi_N}
			&=\tfrac23\mathbb I+o(1),
		\end{aligned}}
			\label{eq:cubic-barycenter}
		\end{equation}
		so a cubic-symmetry-preserving locking texture maps to the barycenter, not to a generic interior point.
		\end{corollary}

		This is the exact tensor consequence of full cubic invariance within the controlled locking description. A non-barycentric simplex point in a cubic material instead diagnoses symmetry lowering by the superconducting order or locking texture, selection of a particular domain by strain or field, or another anisotropic perturbation. Equal populations of symmetry-related domains can restore the barycenter after averaging.

		\begin{corollary}[Directional deficit along $\lbrack111\rbrack$]
		\label{cor:cubic-111}
		For the geometric directional deficit $q(\hat{\mathbf H})\equiv\hat{\mathbf H}\cdot\Pi\cdot\hat{\mathbf H}$, the $[111]$ direction gives
		\begin{equation}
		q_{[111]}
		=\tfrac13\!\left[
		\mathrm{Tr}\,\Pi
		+2(\Pi_{xy}+\Pi_{yz}+\Pi_{zx})
		\right].
		\label{eq:cubic-111-general}
		\end{equation}
		Thus $q_{[111]}=1/3$ follows from $\mathrm{Tr}\,\Pi=1$ only when $\Pi$ is diagonal in the cubic axes; off-diagonal elements contribute otherwise. Full cubic invariance is the stronger statement Eq.~\eqref{eq:cubic-barycenter}.
		\end{corollary}

		As a conditional illustration, the measurements reported in Table~II of Ref.~\cite{Minami2026UBe13} do not violate Corollary~\ref{cor:cubic-barycenter}. They were performed in finite fields that may select a domain or mix order-parameter components, whereas the corollary requires full cubic invariance of the superconducting state and locking texture. Under the additional fixed-tensor, diagonal, single-domain axial ansatz, the data define a putative weakly prolate ellipsoid along $[001]$; they do not establish a zero-field tensor. Appendix~\ref{app:ube13-construction} gives the numerical reconstruction and finite-field qualifications.
		
	\subsection{Hyperfine mapping and intermediate-SOC contraction}
	\label{sec:ellipsoid-hf}
	
	For a diagonal hyperfine tensor $A^{(\mathbf{R})}=\mathrm{diag}(A_{1},A_{2},A_{3})$ coaxial with $\chi_{\mu\nu}$ and $\Pi_{\mu\nu}$,
	\begin{equation}
		K_{i}^{\rm spin}(0)=A_{i}\chi_{N}[1-\pi_{i}+o(1)],\qquad i=1,2,3,
		\label{eq:Khf-ellipsoid}
	\end{equation}
	where $K_i^{\rm spin}$ denotes the spin shift after orbital subtraction. Thus anisotropic hyperfine coupling rescales the principal spin shifts but preserves their normalized ratios. The controlled locking relations become
	\begin{align}
		\sum_{i}\frac{K_{i}^{\rm spin}(0)}{K_{i}^{\rm spin}(T_{c}^{+})} &= 2+o(1), \label{eq:Khf-trace}\\
		\frac{K_{i}^{\rm spin}(0)}{K_{i}^{\rm spin}(T_{c}^{+})} &= 1-\pi_{i}+o(1),\qquad 0\le1-\pi_i\le1. \label{eq:Khf-bound}
	\end{align}
	The remainders vanish identically in the exact zero-SOC unitary-triplet limit. Equation~\eqref{eq:Khf-bound} shows that each $A_{i}$ cancels axis by axis between numerator and denominator: under the coaxial assumption, the normalized principal-value ratios, and therefore the simplex point $(\pi_{1},\pi_{2},\pi_{3})$ extracted from them, are hyperfine-invariant up to the controlled corrections. No knowledge of the magnitudes $A_{i}$ is then required. This makes the geometric procedure transferable across nuclei and compounds; different NMR sites in the same material must return the same simplex point only when they also probe a common electronic locking kernel (Protocol~E).
		If the hyperfine and susceptibility principal frames are not coaxial, this cancellation fails and an independent determination of $A$ is required. We do not pursue that case here.
	
	In the intermediate-SOC $s$-wave case of Eq.~\eqref{eq:chi-explicit}, the scalar-hyperfine Knight shift reads
	\begin{equation}
	\begin{split}
	\frac{K^{\rm spin}(\hat{\mathbf{H}},0)}{K_{N}^{\rm spin}}&=\big\langle F_{s}(\lambda_{\mathbf{k}})\big[1-(\hat{\mathbf{H}}\cdot\hat{\mathbf n}_{\mathbf{k}})^{2}\big]\big\rangle_{\rm FS}\\
	&\equiv S-\hat{\mathbf{H}}\cdot\Pi^{\rm eff}\cdot\hat{\mathbf{H}},
	\end{split}	\label{eq:K-intermediate}
	\end{equation}
	with $S\equiv\langle F_{s}(\lambda_{\mathbf{k}})\rangle_{\rm FS}$ and
	\begin{equation}
		\Pi^{\rm eff}_{\mu\nu}\equiv\langle F_{s}(\lambda_{\mathbf{k}})\hat n_{\mu}\hat n_{\nu}\rangle_{\rm FS},\quad \mathrm{Tr}\,\Pi^{\rm eff}=S\le 1.
		\label{eq:Pi-eff}
	\end{equation}
	Let $p_i$ denote the eigenvalues of $\Pi^{\rm eff}$. Because $\kappa(0)=S\mathbb I-\Pi^{\rm eff}$, they satisfy $p_i=S-\kappa_i(0)$ and $\sum_i p_i=S\le1$.
		For uniform SOC magnitude, $\Pi^{\rm eff}=F_{s}(\lambda)\Pi$, so $p_{i}=F_{s}(\lambda)\pi_{i}$: the strong-SOC ellipsoid contracts homothetically toward the origin as the SOC weakens. In the full three-dimensional map, this changes $S$ while preserving the normalized simplex direction $p_i/S=\pi_i$.
	
	The next section turns these measurable coordinates into an ordered measurement program.
	
	\section{Practical manual for experimentalists}
	\label{sec:guide}
	\label{sec:protocols}
	\label{sec:guide-physics}
	
	The practical task is to reconstruct the orbital-subtracted spin Knight-shift tensor and determine which conclusions survive the required hyperfine, field, and controlled-response checks. Only after the trace condition has been verified do the normalized deficits define a locking geometry; they do not determine pairing parity or a unique order parameter. Tables~\ref{tab:diagnose} and~\ref{tab:protocols} provide the geometric dictionary and measurement checklist.
	
		The minimum workflow consists of the following steps: (i) measure the normal-state and low-temperature spin shifts along three orthogonal axes; (ii) test the trace with Protocol~A; (iii) if sufficient angular data are available, reconstruct the symmetric tensor with Protocol~B; and (iv) add finite-field, $1/T_1$, site-resolved, or pocket-model tests only when the corresponding assumptions are independently controlled. A powder average is insufficient unless isotropic normal-state response and scalar hyperfine coupling justify identifying it with one third of the trace.
	
	\subsection*{Interpretation and warning flags}
	\label{sec:guide-decision}
	\label{sec:guide-diagnose}
	\label{sec:guide-disambiguate}
	
			Use Table~\ref{tab:diagnose} only after orbital subtraction, hyperfine calibration, and the trace check. A simplex point fixes the geometry of the suppressed spin channel, not pairing parity; singlet, parity-mixed, and unitary-triplet states can share it. A trace mismatch rejects the strong-locking identification but does not diagnose non-unitarity. Disagreement between resolved sites rejects at least one of the common-kernel, orbital, or hyperfine assumptions, while failure of a finite-field model requires revising the candidate gap, texture, or parameters. A sheet-resolved interpretation requires the decoupled-pocket assumptions of Sec.~\ref{sec:multiband}; without them, an extremal aggregate point does not determine individual pocket tensors.

		For the trace diagnostic, define
		\[
			r_\alpha\equiv
			\frac{\chi_{\alpha\alpha}(0)}
			{\chi_{\alpha\alpha}(T_c^+)},
			\qquad
			R\equiv\frac{1}{3}\sum_\alpha r_\alpha.
		\]
		In the reference-Fermi-surface normal-state benchmark, $\chi_{\alpha\alpha}(T_c^+)=\chi_N+o(\chi_N)$, so $R=\mathrm{Tr}\,\chi(0)/(3\chi_N)+o(1)$. Equivalently, let $\kappa_i$ be the principal values of $\chi(0)/\chi_N$ and define $q_i\equiv1-\kappa_i$; then $R=\tfrac13\sum_i\kappa_i+o(1)$. If the calibrated electronic normal-state response is appreciably anisotropic, the unweighted mean of directional ratios is not this trace diagnostic.

		\begin{table*}[!t]
		\caption{Microscopic interpretations of measured locking geometries. For scalar hyperfine coupling, principal spin shifts are normalized to $K_N^{\rm spin}$; for a diagonal anisotropic tensor coaxial with $\chi$, normalize each axis by $K_i^{\rm spin}(T_c^+)$. Strong locking gives $\sum_i\pi_i=1$. Intermediate-SOC $s$-wave response gives $\sum_i p_i=S<1$ and, for scalar coupling, $K_i^{\rm spin}(0)/K_N^{\rm spin}=S-p_i$. The final row is a controlled-regime diagnostic, not a universal bound.}
		\label{tab:diagnose}
		\small
		\renewcommand\arraystretch{1.45}
		\setlength{\tabcolsep}{1.2ex}
			\begin{tabular*}{\textwidth}{@{\extracolsep{\fill}} p{0.20\textwidth} p{0.24\textwidth} p{0.22\textwidth} p{0.28\textwidth} @{}}
			\hline\hline
			\rcol{Coordinates/test} & \rcol{Geometry; normalized shifts} & \rcol{Representative realization} & \rcol{Alternatives/limits}\\
			\hline
			\rcol{$p=(0,0,0)$, $S=0$} & \rcol{Collapsed; $(0,0,0)$} & \rcol{Fully gapped singlet, negligible SOC} & \rcol{Reliable orbital subtraction required; residual quasiparticles or impurities prevent exact collapse}\\
			\rcol{$\pi=(\tfrac{1}{3},\tfrac{1}{3},\tfrac{1}{3})$} & \rcol{Sphere; $(\tfrac{2}{3},\tfrac{2}{3},\tfrac{2}{3})$} & \rcol{Strong-SOC singlet with $\Pi^{(g)}=\mathbb I/3$} & \rcol{Parity mixture with the same $\Pi^{(g)}$, or isotropic unitary triplet; cubic crystal symmetry alone is insufficient}\\
			\rcol{$\pi=(0,0,1)$} & \rcol{Oblate disc; $(1,1,0)$} & \rcol{Exact zero-SOC unitary OSP triplet with $\mathbf d\parallel\hat z$} & \rcol{Strong-SOC singlet or parity mixture with quasi-one-dimensional $\hat g\parallel\hat z$; non-unitary states are outside this baseline}\\
			\rcol{$\pi=(\tfrac{1}{2},\tfrac{1}{2},0)$} & \rcol{Prolate axial; $(\tfrac{1}{2},\tfrac{1}{2},1)$} & \rcol{Exact zero-SOC unitary ESP triplet with planar-isotropic $\Pi^{(d)}$} & \rcol{Strong-SOC singlet or parity mixture with ideal two-dimensional Rashba $\Pi^{(g)}$; generic planar textures need not lie at the midpoint}\\
			\rcol{Generic $\pi$, $\sum_i\pi_i=1$} & \rcol{Generally triaxial; $(1-\pi_1,1-\pi_2,1-\pi_3)$} & \rcol{Strong-SOC singlet with anisotropic $\Pi^{(g)}$} & \rcol{Parity mixture sharing $\Pi^{(g)}$; exact zero-SOC unitary triplet with $\Pi^{(d)}=\Pi$; multiband or lower-symmetry texture}\\
			\rcol{Generic $p$, $\sum_i p_i=S<1$} & \rcol{Contracted; $(S-p_1,S-p_2,S-p_3)$} & \rcol{Fully gapped intermediate-SOC $s$ wave with $\Pi^{\rm eff}$} & \rcol{Nonuniform SOC precludes single-$\lambda$ inversion; gap anisotropy or multiband weighting can mimic contraction}\\
				\rcol{$R>2/3\iff\sum_iq_i<1$} & \rcol{$\sum_i\kappa_i>2$} & \rcol{Outside both strong-locking and intermediate-SOC $s$-wave baselines} & \rcol{Can occur for unitary OSP pairing at intermediate SOC; recheck orbital/hyperfine calibration; not a non-unitarity test}\\
			\hline\hline
			\end{tabular*}
	\end{table*}
			
	\subsection*{Consistency checks for orbital-shift subtraction}
	\label{sec:guide-Korb}
	
	The protocols require the spin Knight-shift tensor, obtained from
	\begin{equation}
		K^{\rm meas}_{\mu\nu}(T)=K^{\rm orb}_{\mu\nu}+K^{\rm spin}_{\mu\nu}(T).
		\label{eq:Kdecomp}
	\end{equation}
	When relevant interband splittings $W$ greatly exceed $\Delta$ and no low-energy orbital reconstruction occurs, $K^{\rm orb}_{\mu\nu}$ is approximately temperature- and pairing-independent to leading order in $\Delta/W$. The standard Clogston--Jaccarino procedure~\cite{Clogston1964,Slichter1990,tinkham} extracts it as the intercept of $K^{\rm meas}(T)$ plotted against the corresponding bulk spin susceptibility above $T_c$. Impurity Curie tails, band-dependent hyperfine couplings, or mismatched local and bulk anisotropies can invalidate this extraction. Three internal checks follow.
	
	(i) \emph{Strong-locking trace consistency.} For scalar hyperfine coupling, a candidate strong-locking interpretation requires
	\begin{equation}
		\begin{aligned}
			\sum_\mu \big[K^{\rm meas}_{\mu\mu}(0)-K^{\rm orb}_{\mu\mu}\big]
			&=\tfrac{2}{3}\sum_\mu \big[K^{\rm meas}_{\mu\mu}(T_c^+)-K^{\rm orb}_{\mu\mu}\big]\\
			&\quad+o(K_N^{\rm spin}).
		\end{aligned}
		\label{eq:Korb-tracecheck}
	\end{equation}
	For a diagonal anisotropic hyperfine tensor coaxial with $\chi$, the corresponding test is instead the axis-normalized relation in Eq.~\eqref{eq:Khf-trace}; the raw-shift trace above does not apply. Failure indicates either an inconsistent orbital subtraction or breakdown of the strong-locking assumptions, but does not distinguish the two.
	
	(ii) \emph{Self-consistent $K_N^{\rm spin}$ from temperature differences.} For scalar hyperfine coupling, define $\Delta K_i=K_i^{\rm meas}(T_c^+)-K_i^{\rm meas}(0)$, which cancels $K_i^{\rm orb}$ exactly. The controlled strong-locking trace relation gives
	\begin{equation}
		K_N^{\rm spin}=\sum_i\Delta K_i+o(K_N^{\rm spin}).
		\label{eq:Korb-spinN}
	\end{equation}
	Thus the superconducting transition determines the scalar normal-state spin shift without a normal-state Korringa fit, and $K_i^{\rm orb}=K_i^{\rm meas}(T_c^+)-K_N^{\rm spin}$. The remainder vanishes in the exact zero-SOC unitary-triplet limit. For anisotropic hyperfine coupling, the three $K_i^{\rm spin}(T_c^+)$ are generally unequal, so temperature differences alone do not determine them without additional hyperfine calibration or symmetry constraints.
	
	(iii) \emph{Simplex constraint on the principal projections.} For scalar hyperfine coupling, or a diagonal anisotropic tensor coaxial with $\chi$, define
	\begin{equation}
		\tilde\pi_i(K^{\rm orb})\equiv 1-\frac{K^{\rm meas}_i(0)-K^{\rm orb}_i}{K^{\rm meas}_i(T_c^+)-K^{\rm orb}_i}
		\label{eq:Korb-pitilde}
	\end{equation}
	In the controlled strong-locking regime,
	\begin{equation}
		\sum_i\tilde\pi_i(K^{\rm orb})=1+o(1),
		\label{eq:Korb-simplex}
	\end{equation}
	with exact equality in the zero-SOC unitary-triplet limit. This supplies one scalar consistency condition on a candidate orbital tensor, but does not generally determine all its components. For scalar hyperfine coupling, Eq.~\eqref{eq:Korb-spinN} additionally determines $K_N^{\rm spin}$ and hence each $K_i^{\rm orb}$; for anisotropic hyperfine coupling, further calibration or symmetry constraints are required.
	
	Together, Eqs.~\eqref{eq:Korb-tracecheck}--\eqref{eq:Korb-simplex} make orbital subtraction a joint normal- and superconducting-state consistency problem. Equation~\eqref{eq:Korb-spinN} provides a scalar-coupling alternative to a contaminated normal-state fit; anisotropic coupling requires additional calibration or symmetry. Section~\ref{sec:K2Cr3As3} applies these checks to the Korringa-derived orbital shifts of Ref.~\cite{Triplet2021} only under an explicitly stated fixed-tensor zero-field extrapolation.
	
	\subsection*{Protocol checklist}
		
		Table~\ref{tab:protocols} summarizes the core Protocols~A--B and the conditional Protocols~C--F. Here $(a,b,c)$ denotes crystallographic axes corresponding to the generic $(x,y,z)$ of the theory sections.
	
	\begin{table*}[tb]
			\caption{Six experimental protocols. ``Three-axis $K$'' means orbital-subtracted spin Knight shifts along three orthogonal crystal axes; ``Rotation'' means sufficient symmetry-appropriate orientations to reconstruct $K_{\mu\nu}^{\rm spin}$.}
		\label{tab:protocols}
		\small
		\renewcommand\arraystretch{1.5}
		\setlength{\tabcolsep}{1.0ex}
		\centering
			\begin{tabular*}{\textwidth}{@{\extracolsep{\fill}} p{0.13\textwidth} p{0.16\textwidth} p{0.20\textwidth} p{0.20\textwidth} p{0.21\textwidth} @{}}
			\hline\hline
			\rcol{Protocol} & \rcol{Measurement} & \rcol{Output/test} & \rcol{Required conditions} & \rcol{Outcome/limit}\\
			\hline
				\rcol{A: Trace diagnostic} & \rcol{Three-axis $K^{\rm spin}(T)$} & \rcol{Test $R=\tfrac{1}{3}\sum r_{\alpha}$ against $2/3$} & \rcol{Calibrated hyperfine; isotropic reference-FS normal response; controlled-response regime} & \rcol{$R\ne2/3$ rejects strong locking; $R>2/3$ also rejects the intermediate-SOC $s$-wave baseline}\\
			\rcol{B: Ellipsoid} & \rcol{Rotation at $T_c^+$ and $T\to0$} & \rcol{$q_i$; intermediate-SOC $p_i$} & \rcol{Scalar or calibrated coaxial hyperfine; strong locking to set $q_i=\pi_i$} & \rcol{Require $\sum_iq_i=1$ before simplex mapping}\\
			\rcol{C: Finite-$H$ response} & \rcol{Rotation vs.\ $H$ at $T\to0$} & \rcol{Angular shift surface $K^{\rm spin}(\hat{\mathbf H},0,H)$} & \rcol{Candidate-specific self-consistent $\mathbf M(\mathbf H)$} & \rcol{An ellipsoid exists only if one tensor fits the angular data}\\
			\rcol{D: $1/T_{1}$ limitation} & \rcol{$1/T_{1}(T,\hat{\mathbf{H}})$} & \rcol{Momentum-summed transverse response; conditional $\mathbf q=0$ exclusion} & \rcol{Identified $\Pi^{(g)}$; both transverse components zero; weak field; calibrated hyperfine} & \rcol{Cannot separate small nonzero $\mathbf q$ from finite wave vector}\\
			\rcol{E: Site-resolved} & \rcol{Site-resolved $K^{\rm spin}$} & \rcol{Compare $(\pi_i)^{(\mathbf{R})}$ across sites} & \rcol{Common electronic kernel; scalar or calibrated coaxial hyperfine; strong locking} & \rcol{Disagreement rejects the joint assumptions without identifying the cause}\\
			\rcol{F: Baseline test} & \rcol{$K^{\rm spin}$ vs.\ pressure/composition} & \rcol{$w_{\parallel}$ and full-tensor trajectory} & \rcol{Specified strong-locking decoupled-pocket baseline} & \rcol{Departure from Eq.~\eqref{eq:a2cr3as3-pred} rejects only that baseline}\\
			\hline\hline
			\end{tabular*}
	\end{table*}
	
	\subsection{Protocol A: Trace diagnostic from three-axis spin shifts}
	\label{sec:protocolA}
	
			Protocol~A is the minimal-data trace test within the isotropic reference-Fermi-surface normal-state benchmark. It distinguishes controlled trace regimes but neither supplies the full locking geometry nor defines a universal susceptibility bound.
	
		\begin{enumerate}
			\item Measure $K_{\alpha}^{\rm meas}$ for $\alpha=a,b,c$ at $T_c^+$ and $T\to0$; subtract $K_{\alpha}^{\rm orb}$ and apply the calibrated hyperfine tensor to obtain $\chi_{\alpha\alpha}$.
				\item Verify that the calibrated normal-state response is consistent with $\chi_{\alpha\alpha}(T_c^+)=\chi_N+o(\chi_N)$ for a common $\chi_N$, then form $r_{\alpha}\equiv\chi_{\alpha\alpha}(0)/\chi_{\alpha\alpha}(T_{c}^{+})$. If appreciable electronic normal-state anisotropy remains, the unweighted mean below is not the controlled trace diagnostic without a model-specific normalization.
			\item Compute $R=\tfrac{1}{3}\sum_{\alpha}r_{\alpha}$.
			\item Compare $R$ with $2/3$. Strong locking requires $R=2/3+o(1)$; the intermediate-SOC $s$-wave model permits $R\le2/3$. Thus $R>2/3$ rejects both baselines but not unitary pairing.
			\item Use $(r_a,r_b,r_c)$ only to compare the controlled benchmarks in Table~\ref{tab:sat}, not to assign a unique pairing class.
		\end{enumerate}
	
	\subsection{Protocol B: Locking-ellipsoid reconstruction}
	\label{sec:protocolB}
	
			Protocol~B reconstructs the locking ellipsoid from sufficient orientation data appropriate to the crystal symmetry. Its geometric output provides the input to Protocols~C--F.
	
	\begin{enumerate}
		\item After orbital subtraction, reconstruct $K_{\mu\nu}^{\rm spin}(T_c^+)$ from symmetry-appropriate orientations in the linear-response regime.
		\item Repeat under the same field conditions at $T\to0$ to obtain $K_{\mu\nu}^{\rm spin}(0)$.
		\item Extract the principal axes and normalized responses $\kappa_i$: use $\kappa_i=K_i^{\rm spin}(0)/K_N^{\rm spin}$ for scalar hyperfine coupling, or $\kappa_i=K_i^{\rm spin}(0)/K_i^{\rm spin}(T_c^+)$ for a diagonal hyperfine tensor coaxial with $\chi$.
			\item Define $q_i\equiv1-\kappa_i$. Only if $q_i\ge0$ and $\sum_iq_i$ is consistent with $1$ within the experimental uncertainty and controlled corrections may one identify $q_i=\pi_i$ and place $(\pi_1,\pi_2,\pi_3)$ on the simplex of Fig.~\ref{fig:simplex}.
			\item For intermediate-SOC $s$-wave data, define $S=\tfrac{1}{2}\sum_i\kappa_i$ and $p_i=S-\kappa_i$. Require $S\le1$ and $p_i\ge0$; the relation $\sum_i p_i=S$ then follows algebraically. Place $(p_1,p_2,p_3)$ in the tetrahedral map of Fig.~\ref{fig:simplex}(c). For $S>0$, record $S$ as the radial coordinate and plot $(p_1/S,p_2/S,p_3/S)$ as the normalized direction on the simplex. Uniform SOC preserves this direction as $S$ varies; nonuniform SOC need not. At $S=0$, the response is collapsed and no direction is defined.
		\item Interpret the result as a locking geometry only. Identifying a microscopic pairing state requires the additional tests of Protocols~C--F.
	\end{enumerate}
		The number of independent orientations required at each temperature is set by the number of symmetry-allowed components of $K_{\mu\nu}^{\rm spin}$: two for uniaxial, three for orthorhombic, four for monoclinic, and six for triclinic symmetry. For a uniaxial crystal, $\mathbf H\parallel\hat c$ and $\mathbf H\perp\hat c$ yield $(q_\perp,q_\perp,q_\parallel)$ without a full angular scan. Only under the controlled strong-locking identification must $2q_\perp+q_\parallel=1+o(1)$; this is not a general symmetry relation. The two orientations therefore both reconstruct the axial geometry and test strong locking. This counting assumes that one field-independent tensor describes all orientations; otherwise use Protocol~C.
	
	In the uniform-gap intermediate-SOC $s$-wave model, Protocol~B gives
	\[
		S=\frac{\mathrm{Tr}\,\chi(0)}{2\chi_N}
		=\langle F_s(\lambda_{\mathbf k})\rangle_{\rm FS},
		\qquad \lambda_{\mathbf k}=|\mathbf g_{\mathbf k}|/\Delta.
	\]
	For uniform $\lambda_{\mathbf k}$, monotonicity gives $\lambda=F_s^{-1}(S)$. Otherwise the inversion defines only $\lambda_{\rm eff}$ in Eq.~\eqref{eq:lambda-eff}; Eq.~\eqref{eq:Fs-variance} gives its narrow-distribution correction. Extracting the mean SOC-to-gap ratio requires independent information about the distribution of $|\mathbf g_{\mathbf k}|$, or dominance by one near-isotropic sheet.
	
	\subsection{Protocol C: Empirical finite-field response}
	\label{sec:protocolC}
	
		Protocol~C maps the finite-field spin-shift surface empirically. Because a finite Zeeman field can invalidate the zero-field helicity reduction and reorient the order parameter, do not infer a field-dependent locking tensor; Appendix~\ref{app:finite-field-response} defines the measured secant response and its relation to the differential Kubo response.
	
	\begin{enumerate}
		\item Measure $K^{\rm spin}(\hat{\mathbf H},T\to0,H)$ over sufficient orientations at each common field magnitude, with orbital and diamagnetic contributions treated consistently.
		\item Test whether one symmetric tensor $K_H^{\rm eff}$ reproduces the angular data through $K^{\rm spin}(\hat{\mathbf H},0,H)=\hat{\mathbf H}\cdot K_H^{\rm eff}\cdot\hat{\mathbf H}$. If it does, track the resulting field-conditioned effective ellipsoid. If it does not, retain the measured angular surface without assigning principal axes or simplex coordinates.
		\item For each candidate gap and SOC texture, solve the finite-field BdG problem self-consistently and calculate the equilibrium magnetization $\mathbf M(T,\mathbf H)$. Compare the measured shift with $\hat{\mathbf H}\cdot A\mathbf M/H$ at the actual experimental field vector. Use a Kubo differential susceptibility directly only in linear response or after relating it to $\mathbf M/H$.
		\item Compare the angular and field evolution, including any nonanalytic features, directly with the calculation. Do not equate finite-field deficits with eigenvalues of the normal-state tensor $\Pi^{(b)}(\mathbf H)$. The cancellation field $H^\ast$ in Eq.~\eqref{eq:Hstar} is only a normal-state scale; Eq.~\eqref{eq:BFS-onset} is one model-specific superconducting onset condition.
	\end{enumerate}
	
	\subsection{Protocol D: Momentum limitation of $1/T_{1}$}
	\label{sec:protocolD}
	
		Measure the transverse relaxation rate with $\mathbf H$ aligned along each principal axis $\hat e_\alpha$, and retain the full momentum-summed interpretation. Ordinary $1/T_1$ data do not decompose the rate into exact $\mathbf q=0$ and nonzero-$\mathbf q$ parts, and they cannot distinguish near-uniform fluctuations from a response peaked at a finite wave vector.
	
		Only when independent evidence identifies the tensor as $\Pi^{(g)}$, both transverse projections vanish, and the weak-field and calibrated-hyperfine conditions hold may Eq.~\eqref{eq:T1-vanishing} be used to exclude the exact $\mathbf q=0$ term. The remaining rate can still be dominated by arbitrarily small nonzero $\mathbf q$. Appendix~\ref{app:dynamical-response} gives the derivation, Kramers--Kronig distinction, and Hebel--Slichter qualification.
	
	\subsection{Protocol E: Consistency across resolved NMR sites}
	\label{sec:protocolE}
	
		Protocol~E tests whether different resolved NMR sites probe a common electronic locking kernel. The normalized ratios of Sec.~\ref{sec:ellipsoid-hf} can be compared directly only for scalar or diagonal coaxial hyperfine coupling; otherwise the tensor must be calibrated independently. Equation~\eqref{eq:site-budget} gives an exact site-resolved operator budget but does not guarantee a site-independent normalized response.
	
	\begin{enumerate}
		\item Identify all resolved NMR-active sites for which $K^{\rm spin}$ can be extracted with adequate signal-to-noise and a reliable orbital subtraction.
			\item Determine the hyperfine tensor at each site. Apply the normalized Protocol~B construction directly only for scalar coupling or a diagonal tensor coaxial with $\chi$ and $\Pi$; otherwise calibrate the full tensor and reconstruct the electronic response before comparing sites. At each site, identify $(\pi_i)^{(\mathbf R)}=(q_i)^{(\mathbf R)}$ only after verifying $(q_i)^{(\mathbf R)}\ge0$ and $\sum_i(q_i)^{(\mathbf R)}=1+o(1)$.
			\item Compare $(\pi_i)^{(\mathbf R)}$ across sites, with the principal axes assigned consistently in the crystallographic basis and with statistical, orbital-subtraction, and hyperfine-calibration uncertainties propagated. Agreement within the combined uncertainty supports the common-kernel and calibration assumptions. Treat approximately $5\%$ componentwise precision as an experimental target, not a theoretical threshold.
				\item A significant discrepancy shows that at least one joint assumption has failed. Possible causes include interband structure or non-helicity-diagonal pairing~\cite{Frigeri2004,bauer2012non,Samokhin2007}, non-unitary response, site-specific orbital shifts, or incomplete hyperfine calibration. The magnitude of the discrepancy alone cannot distinguish these causes; resolving them requires site-specific calibration and dedicated multiorbital response modeling.
	\end{enumerate}
	
	\subsection{Protocol F: Decoupled-pocket baseline test}
	\label{sec:protocolF}
	
			Protocol~F is the final, model-dependent layer. After Protocols~A and~B verify the trace condition and reconstruct the locking geometry, test the specified decoupled-pocket model of Sec.~\ref{sec:multiband}; for A$_2$Cr$_3$As$_3$, it reduces to Eqs.~\eqref{eq:a2cr3as3-perp}--\eqref{eq:a2cr3as3-par}.
	
		\begin{enumerate}
			\item Within this baseline, infer $w_{\parallel}=[\chi_{xx}(0)-\chi_{zz}(0)]/\chi_N+o(1)$ from the calibrated output of Protocol~A or B, with $x$ a basal-plane principal axis. This is a model-derived normal-state susceptibility weight and need not equal a bare density-of-states fraction.
			\item Check the trace relation $\sum_\mu\chi_{\mu\mu}(0)=[2+o(1)]\chi_N$. The decoupled-pocket baseline predicts this relation when pairing and spin-response vertices are block diagonal and every relevant pocket lies in the controlled strong-locking regime. Failure rejects the joint baseline; agreement is necessary but not sufficient to establish its pocketwise assumptions.
			\item Under pressure or composition tuning, repeat the calibrated tensor measurement and infer $w_{\parallel}$ at each point using Step~1. Test whether $\chi_{zz}(0)/\chi_N$ follows Eq.~\eqref{eq:a2cr3as3-par} only over the range where the superconducting phase, assigned pocket textures, and pocketwise strong-locking conditions remain applicable.
			\item At each tuning point, compare the full measured tensor with Eq.~\eqref{eq:a2cr3as3-pred} using the inferred $w_{\parallel}$. Agreement within the propagated uncertainty supports, but does not establish, the specified weights and SOC textures within the decoupled-pocket baseline. Systematic disagreement rejects only that joint baseline; it neither selects a common $\mathbf d$ vector nor rules out all band-resolved SOC mechanisms.
		\end{enumerate}
		If the block-diagonal pocket assumptions fail, the exact Bogoliubov sum rule remains valid in the full spin-orbital space but no pocketwise response identity follows; interpretation then requires a multiorbital calculation of the type discussed in Sec.~\ref{sec:multiband}.
	
	\section{Application to K$_{2}$Cr$_{3}$As$_{3}$}
	\label{sec:K2Cr3As3}
	
	We apply the protocols of Sec.~\ref{sec:protocols} to the finite-field $^{75}$As Knight-shift and $1/T_1$ data of Yang \emph{et al.}~\cite{Triplet2021}. We separate the empirical anisotropic suppression from the additional fixed-tensor and zero-field assumptions needed for a simplex interpretation, then test a specified decoupled-pocket SOC-texture baseline and discuss microscopic scenarios and follow-up measurements.
	
	\subsection{Finite-field Knight-shift constraints and conditional $1/T_1$ test}
	\label{sec:K2Cr3As3-robust}
	
	Reference~\cite{Triplet2021} reports the spin Knight shift at the As2 site using a Korringa-based orbital subtraction, with $K_{\rm orb}^{c}=0.27\%$ and $K_{\rm orb}^{\perp c}=0.09\%$. The measurements were not performed in the zero-field linear-response limit. For $\mathbf H\perp\hat c$, Fig.~5(a) uses $H=8.4949$~T. For $\mathbf H\parallel\hat c$, Fig.~5(b) uses $H=8.4925$, $9.6955$, $11.9990$, and $16.0454$~T, while Fig.~6 compares the $11.9990$~T $c$-axis spectrum with the $8.4949$~T in-plane spectrum. Define the empirical field-conditioned deficit
	\begin{equation}
		q_{\alpha}(T,H)\equiv
		1-\frac{K_{\alpha}^{\rm spin}(T,H)}
		{K_{\alpha}^{\rm spin}[T_c^{+}(H),H]}.
		\label{eq:K2Cr3As3-raw}
	\end{equation}
	The data give $q_{\perp c}\simeq0$ down to the lowest measured temperature at $8.4949$~T, whereas $q_c$ becomes positive below a field-dependent onset $T^\ast(H)$ for the lower $c$-axis fields. No $c$-axis reduction is resolved at $16.0454$~T. Reference~\cite{Triplet2021} extrapolates the lower-field $c$-axis spin shift toward zero as $T\to0$ and attributes the field dependence to order-parameter depinning, with a pinning scale no larger than about $13$~T. The sixfold in-plane symmetry supports an axially symmetric in-plane response but does not establish one field-independent tensor over polar angles.
	
		(O1) \emph{Empirical finite-field observation.} After the adopted orbital subtraction and calibrated hyperfine mapping, the data show no resolved in-plane suppression but do show a field-dependent $c$-axis suppression. They do not by themselves determine the zero-field locking tensor $\Pi$. If one additionally assumes that the authors' $T\to0$ extrapolation applies near the common field $H\simeq8.5$~T, that a single axial secant-response tensor describes both orientations, and that field-induced order-parameter relaxation is negligible in this limit, then
	\begin{equation}
		(q_a,q_b,q_c)_{H\simeq8.5\,{\rm T}}^{\rm cond}
		=(0\pm\epsilon_{ab},\,0\pm\epsilon_{ab},\,1\pm\epsilon_c).
		\label{eq:K2Cr3As3-pi}
	\end{equation}
	Here $\epsilon_{ab}$ and $\epsilon_c$ include the experimental and extrapolation uncertainties. Only with the further approximation $\bar\chi(T;\mathbf H)\simeq\chi^{\rm diff}(T;\mathbf0)$ and the controlled zero-field strong-locking identity may one identify the conditional deficits in Eq.~\eqref{eq:K2Cr3As3-pi} with $\pi_i$ and place the inferred point near the $(0,0,1)$ vertex of Fig.~\ref{fig:simplex}. The sum $\sum_iq_i\simeq1$ is then a consistency check on this extrapolation, not a measured zero-field trace identity. A sheet-resolved inference additionally requires an additive decoupled-pocket decomposition; without it, the aggregate finite-field anisotropy does not determine individual pocket tensors.
	
		(C1) \emph{Conditional selection-rule consequence.} Even the conditional $\hat c$-vertex placement does not by itself imply Eq.~\eqref{eq:T1-vanishing}, because the same zero-field point can represent either $\Pi^{(g)}$ or the zero-SOC unitary-triplet tensor $\Pi^{(d)}$. If independent evidence identifies the extrapolated tensor as $\Pi^{(g)}$, and a finite-field calculation or separate scale argument shows that the actual measurement field preserves the antiparallel helicity structure, the order parameter, and the calibrated transverse hyperfine channels, then
	\begin{equation}
		\begin{gathered}
			\Pi^{(g)}_{aa}=\Pi^{(g)}_{bb}=0\\[-2pt]
			\Longrightarrow\;
			(1/T_{1}T)_{\mathbf{q}=0,\hat c}^{N}
			=(1/T_{1}T)_{\mathbf{q}=0,\hat c}^{\rm SC}=0
		\end{gathered}
		\label{eq:K2Cr3As3-T1vanish}
	\end{equation}
	Equation~\eqref{eq:K2Cr3As3-T1vanish} is therefore a candidate-specific finite-field test, not a deduction from the reported $8$--$12$~T Knight shifts. If its assumptions hold, it excludes only the exact-$\mathbf q=0$ contribution; the remaining $1/T_1$ can still be dominated by arbitrarily small nonzero $\mathbf q$. A zero-SOC unitary-triplet realization $\Pi^{(d)}$ yields no such exclusion from the simplex point alone. The normal-state Curie--Weiss enhancement is consistent with ferromagnetic or small-$\mathbf q$ fluctuations, but $1/T_1$ determines neither their momentum profile nor their role in pairing. The absence of a Hebel--Slichter peak is an independent gap-coherence observation, not unique evidence for nodal $p$-wave pairing.
	
	\subsection{Conditional test of the specified decoupled-pocket SOC baseline}
	\label{sec:K2Cr3As3-rulesout}
	
	The conditional zero-field vertex placement in O1, not the finite-field data alone, constrains an additive multiband interpretation. The specified decoupled-pocket SOC baseline of Sec.~\ref{sec:multiband-A2Cr3As3} assigns quasi-one-dimensional $\hat c$-axis locking to the $\alpha,\beta$ pockets and an isotropic three-dimensional locking texture, $\Pi^{(\gamma)}\simeq\mathbb I/3$, to the $\gamma$ pocket. It therefore predicts
	\begin{equation}
		\begin{split}
			(\pi_{a},\pi_{b},\pi_{c})^{\rm SOC\text{-}texture}
			={}& \Big(\tfrac{w_{\gamma}}{3},\tfrac{w_{\gamma}}{3},\,w_{\parallel}+\tfrac{w_{\gamma}}{3}\Big) \\
			\approx{}& (0.25,\,0.25,\,0.50),
		\end{split}
		\label{eq:K2Cr3As3-pred-SOC}
	\end{equation}
	Taking the DFT-derived $\gamma$-pocket density-of-states fraction $w_\gamma\approx0.75$ quoted in Ref.~\cite{Triplet2021} as a proxy for the normal-state susceptibility weight gives the numerical point in Eq.~\eqref{eq:K2Cr3As3-pred-SOC}. This axial point lies one quarter of the way from the barycenter to the $\hat c$-vertex of Fig.~\ref{fig:simplex}. Its largest-component discrepancy from the conditionally extrapolated vertex is $0.5$; even for $w_\gamma=0.5$, the discrepancy remains approximately $0.33$. Thus the specified baseline is disfavored under the fixed-tensor zero-field extrapolation, but the finite-field data alone do not reject it.
	
	This comparison jointly assumes a field-stable response, block-diagonal pocket additivity, the stated susceptibility weights, strong locking on every pocket, and the assigned SOC textures; order-parameter relaxation, interband pairing, orbital hybridization, or interband spin vertices can invalidate it. Within this restricted setting, a common $\hat c$-axis locking texture remains possible, with a common $\mathbf d$ vector as only one realization. The DFT calculations of Refs.~\cite{Jiang15,Wu15}, which do not find generic $\hat c$-axis alignment of the $\gamma$-pocket SOC texture, further disfavor an SOC-only realization of the specified baseline.
	
	\subsection{Microscopic scenarios compatible with conditional $\hat c$-axis locking}
	\label{sec:K2Cr3As3-open}
	
		The extrapolated $(0,0,1)$ vertex is compatible with multiple mechanisms under distinct assumptions. We label the following illustrative zero-field scenarios S1--S3, where S denotes scenario: S1 is a controlled zero-SOC unitary OSP realization; S2 is a conditional helicity-diagonal strong-SOC construction with $\mathbf d\parallel\mathbf g\parallel\hat c$; and S3 is a non-unitary triplet candidate outside the zero-SOC unitary-triplet baseline.
	
	(S1) \emph{Common-axis unitary OSP triplet, $\hat{\mathbf d}_{\mathbf k}\parallel\hat c$.} This is the simplest controlled zero-field realization within the present unitary BdG framework, not a controlled calculation of the $8$--$16$~T response. At zero SOC, the exact identity Eq.~\eqref{eq:strong-triplet} gives $\Pi^{(d)}=\mathrm{diag}(0,0,1)$ for a common $\hat c$-axis $\mathbf d$ texture on every band. The symmetry analysis of Ref.~\cite{Triplet2021} (Table~1 there) identifies the $E'$ irreducible representation $(p_x\pm ip_y)\hat z$ as the natural triplet candidate compatible with this scenario.
	
	(S2) \emph{Conditional helicity-diagonal parity admixture.} If every appreciably weighted sheet has $\mathbf g_{\mathbf k}\parallel\hat c$ and supports a fully gapped helicity-diagonal state with $\mathbf d_{\mathbf k}\parallel\mathbf g_{\mathbf k}$, Eq.~\eqref{eq:paritymix-T0} gives the same zero-field $(0,0,1)$ tensor as S1. Applying this construction at the experimental fields requires a self-consistent finite-field response calculation. The zero-field Knight-shift tensor fixes the locking geometry, not the singlet--triplet admixture. Two distinct finite-$T$ recovery scales can support split helicity gaps only when both helicity sheets carry appreciable response weight; even then, multiband gaps, gap anisotropy or nodes, and impurity broadening remain alternative explanations (Sec.~\ref{sec:paritymix-consequences}).
	
			(S3) \emph{Non-unitary triplet outside the zero-SOC unitary-triplet baseline.} A non-unitary state has $i\mathbf d_{\mathbf k}\times\mathbf d_{\mathbf k}^{*}\ne0$ and therefore requires at least two noncollinear spin components with a relative complex phase; a $\mathbf d$ vector confined to one real spin axis is unitary. The zero-SOC unitary-triplet identity therefore does not place such a state at the $(0,0,1)$ vertex. S3 remains an external candidate only if an explicit response calculation for a non-unitary and potentially multiband state reproduces the data. The Curie--Weiss enhancement reported in Ref.~\cite{Triplet2021} may motivate triplet physics, but it neither establishes non-unitarity nor makes the absence of a Hebel--Slichter peak diagnostic.
	
	Discriminating among these illustrative scenarios requires probes beyond the static Knight-shift tensor:
	\noindent(a) Time-reversal-symmetry-breaking probes such as zero-field $\mu$SR and polar Kerr rotation test a common consequence of non-unitary pairing; a positive signal would support broken time-reversal symmetry but would not uniquely identify S3.

	\noindent(b) Finite-$T$ Knight-shift recovery, interpreted together with independent constraints on multiband gaps, gap anisotropy or nodes, and impurity broadening, can constrain S1 versus the conditional S2 construction. A resolved two-scale recovery supports split helicity gaps only under the S2 conditions and is not unique evidence for parity admixture.

	\noindent(c) If interband pairing, orbital hybridization, or interband spin-vertex matrix elements are significant, a dedicated multiorbital BdG response calculation is required. The exact Bogoliubov sum rule, Eq.~\eqref{eq:general-sum-rule}, remains valid in the full spin-orbital Hilbert space, but the anisotropic response and pocket identities must be rederived.
	
	\subsection{Follow-up tests}
	\label{sec:K2Cr3As3-predictions}
	
		We label the following proposed follow-up tests P1 and P2. These two measurements would test the working assumptions and help discriminate among the remaining microscopic interpretations.
	
	(P1) \emph{Field and site consistency test.} Perform common-field angular scans and field sweeps at the As1 and As2 sites with site-specific orbital subtraction and hyperfine calibration. First test whether one secant-response tensor describes each common field, then extrapolate toward $H\to0$. Only after each site satisfies $(q_i)^{(\mathbf R)}\ge0$ and $\sum_i(q_i)^{(\mathbf R)}=1+o(1)$ in the controlled regime may the extracted $(\pi_i)^{(\mathbf R)}$ be compared. Treat approximately $5\%$ componentwise precision as the Protocol~E target, not a theoretical threshold. A significant discrepancy rejects at least one joint assumption, including field independence, orbital subtraction, hyperfine calibration, common-kernel description, or neglect of multiorbital and interband effects; its magnitude alone does not identify the cause.
	
	(P2) \emph{Pressure/composition test of competing baselines.} If the S1 common-$\mathbf d$ texture and its controlled zero-field response remain unchanged under tuning, the extrapolated $(0,0,1)$ tensor is independent of the pocket susceptibility weights. In the specified decoupled-pocket SOC baseline, changing $w_\gamma$ instead moves the zero-field tensor along Eq.~\eqref{eq:a2cr3as3-pred}. Repeat the calibrated measurement under pressure or composition tuning at common fields and with an $H\to0$ extrapolation, while checking that the superconducting phase, assigned textures, and hyperfine mapping remain controlled. The extrapolated trajectory can distinguish these two baselines, but it gives no universal prediction for arbitrary band-resolved SOC or interband pairing.
	
	The finite-field data establish anisotropic spin-shift suppression. The $\hat c$-vertex placement and tension with the specified decoupled-pocket baseline additionally require the fixed-tensor $T\to0$ and zero-field extrapolation of O1. P1 tests this field and site consistency, whereas P2 compares the common-$\mathbf d$ and specified decoupled-pocket SOC baselines under tuning. Neither the static tensor nor these follow-ups uniquely identify the pairing state; complementary probes and material-specific finite-field response calculations remain required.
	
	\section{Summary}
	\label{sec:summary}
	
		The exact foundation is the Bogoliubov sum rule: for any Hermitian single-particle operator $O$ and each momentum $\mathbf k$, the particle-hole and particle-particle squared-matrix-element weights sum to $\operatorname{Tr}_{s}(O^{2})$. This BdG-doubled Hilbert--Schmidt-norm identity generates spin, scalar-density, multiorbital, and site-resolved operator budgets. Because it contains neither excitation-energy denominators nor occupations, it supplies only a kinematic numerator budget; by itself, it provides neither a susceptibility bound nor a conservation law for dynamical spectral weight.

	The geometric response identity underlying the diagnostics holds in two controlled zero-field limits: the fully gapped, two-sheet helicity-diagonal strong-SOC regime and the fully gapped zero-SOC unitary-triplet regime:
	\begin{equation*}
		\chi_{\mu\nu}(0)=\chi_N[\delta_{\mu\nu}-\Pi_{\mu\nu}]+o(1).
	\end{equation*}
	Here $\Pi=\langle\hat g_\mu\hat g_\nu\rangle_{\rm FS}$ in the strong-SOC limit and $\Pi=\Pi^{(d)}$ in the zero-SOC unitary-triplet limit, where the $o(1)$ term vanishes identically. In both cases, $\Pi$ is positive semidefinite with $\operatorname{Tr}\Pi=1$ and is not generally a projector. After orbital subtraction, controlled hyperfine mapping, and verification of the trace constraint, its eigenvalues can be identified with the principal Knight-shift deficits. They define the simplex point and ellipsoid geometry, but do not uniquely identify the microscopic pairing mechanism.

	The two controlled locking limits give $\operatorname{Tr}\chi(0)=2\chi_N+o(1)$, with no remainder in the zero-SOC unitary-triplet limit. For zero-field pure $s$-wave pairing at arbitrary $\lambda_{\mathbf k}=|\mathbf g_{\mathbf k}|/\Delta$ within the reference-Fermi-surface regime, the exact interpolation instead gives
	$\operatorname{Tr}\chi(0)/\chi_N=2\langle F_s(\lambda_{\mathbf k})\rangle_{\rm FS}\le2$.
	Fully gapped helicity-diagonal parity mixtures share the corresponding leading $T=0$ tensor geometry, although distinct helicity gaps can produce multiple recovery scales. Full cubic invariance of both the superconducting state and locking texture fixes $\Pi=\mathbb I/3$ exactly and hence $\chi(0)/\chi_N=2\mathbb I/3+o(1)$; cubic crystal symmetry alone does not. In a finite Zeeman field, the zero-field helicity reduction generally fails. A self-consistent BdG calculation must determine the equilibrium magnetization and differential response, while the Knight shift probes the secant response $\mathbf M(\mathbf H)/H$; rotating-field data therefore define a general angular surface unless one tensor fits them.

	Causality yields a signed inverse-frequency Kramers--Kronig relation for the $\mathbf q\to0$ spectral difference, not a spin FGT conservation law, whereas $1/T_1$ samples transverse low-frequency response over all momenta. For an independently identified strong-SOC helicity texture with nonzero $g_{\min}$ and $\Pi^{(g)}_{\beta\beta}=\Pi^{(g)}_{\gamma\gamma}=0$, the exact-$\mathbf q=0$ contribution vanishes under the calibrated-hyperfine and weak-field conditions. This selection rule does not follow from $\Pi^{(d)}$; it neither excludes arbitrarily small nonzero momentum nor provides a model-independent relaxation estimate. Similarly, pocketwise identities require block-diagonal pairing and spin vertices and a controlled response identity on every pocket; interband pairing, hybridization, or interband spin vertices require a full multiorbital calculation. The six protocols enforce these validity checks.
	
	For K$_{2}$Cr$_{3}$As$_{3}$, O1 separates the measured field-dependent axial suppression at $8$--$16$~T from the conditional $(0,0,1)$ zero-field extrapolation. The latter requires a field-independent secant-response tensor, negligible order-parameter relaxation, the controlled strong-locking identity, and pocket additivity for sheet-resolved claims. C1 excludes only the exact-$\mathbf q=0$ contribution and only after independent identification of $\Pi^{(g)}$ and finite-field preservation of the helicity selection rule. Under the same extrapolation, the approximately $0.5$ largest-component mismatch disfavors, but does not experimentally exclude, the specified decoupled-pocket SOC baseline. S1 is a controlled zero-field unitary OSP realization, S2 is a conditional strong-SOC helicity construction, and S3 requires a non-unitary response theory. Common-field angular scans, $H\to0$ extrapolation, site-resolved NMR, and pressure or composition tuning test these assumptions.
	
		Overall, the framework places the Knight shift and $1/T_1$ in one BdG matrix-element setting while keeping the uniform static response, the $\mathbf q\to0$ inverse-frequency relation, and momentum-summed relaxation distinct. Its hierarchy consists of an exact operator budget, controlled response identities, and model-dependent interpretations. The sum rule remains exact for non-unitary and multiorbital BdG Hamiltonians; the anisotropic response and pocketwise identities must be rederived.

	\begin{acknowledgments}
	This work was supported by the National Key R\&D Program of China (Grant No. 2022YFA1403403) and the National Natural Science Foundation of China (Grant Nos. 12274441 and 12534004). The author thanks X.~Lin, Z.~Li, R.~Zhou, T.~Wu, and G.~Q.~Zheng for stimulating discussions and helpful comments.
	\end{acknowledgments}
	
	\section*{Data and Code Availability}
		For transparency, the numerical-validation code and records of the AI-assisted manuscript-preparation process are publicly available in the GitHub repository \texttt{Bogoliubov-Sum-Rule}~\cite{yizhou76-sudo2026bogoliubov}. Following the transparency practice exemplified in the essay \textit{``Co-Authoring with AI: How I Wrote a Physics Paper About AI, Using AI"} [\href{https://arxiv.org/abs/2604.04081}{arXiv:2604.04081}], the repository contains:
	\noindent(i) \textbf{Markdown transcripts:} Records of manuscript-preparation interactions with GPT-5.4, Claude Opus 4.7, and Codex.

	\noindent(ii) \textbf{Validation code:} The Python script for the numerical cross-check and the LaTeX source used to generate Fig.~\ref{fig:fig3-t0-validation}.
	
	\appendix

	\section{Notation}
	\label{app:notation}

	Table~\ref{tab:notation} summarizes the notation used throughout the paper.

		\begin{center}
		\begin{minipage}{\columnwidth}
			\captionof{table}{Key symbols.}
			\label{tab:notation}
			\small
		\renewcommand\arraystretch{1.5}
		\setlength{\tabcolsep}{0.5ex}
			\begin{tabular*}{\columnwidth}{@{\extracolsep{\fill}} p{0.31\columnwidth} p{0.67\columnwidth} @{}}
			\hline\hline
			Symbol & \rcol{Definition} \\
			\hline
			$\mathbf{g}_{\mathbf{k}}$ & \rcol{Antisymmetric SOC field, $\mathbf{g}_{-\mathbf{k}}=-\mathbf{g}_{\mathbf{k}}$} \\
			$\mathbf{b}_{\mathbf{k}}$ & \rcol{$\mu_B\mathbf H+\mathbf g_{\mathbf k}$ (Eq.~\eqref{eq:ndef})} \\
			$\hat{\mathbf n}_{\mathbf{k}}$ & \rcol{$\mathbf b_{\mathbf k}/|\mathbf b_{\mathbf k}|$ (Eq.~\eqref{eq:ndef})} \\
			$\psi(\mathbf{k}),\mathbf{d}(\mathbf{k})$ & \rcol{Singlet and triplet pairing amplitudes (Eq.~\eqref{eq:Delta-gen})} \\
			$\lambda_{\mathbf{k}}$ & \rcol{$|\mathbf g_{\mathbf k}|/\Delta$ for pure $s$-wave pairing} \\
			$h$ & \rcol{Source field coupled as $-h\hat O$} \\
			$\Pi_{\mu\nu}$ & \rcol{Locking tensor $\langle\hat n_\mu\hat n_\nu\rangle_{\rm FS}$, or $\Pi_{\mu\nu}^{(d)}$ in the zero-SOC unitary-triplet limit (Eqs.~\eqref{eq:Pidef} and~\eqref{eq:strong-triplet})} \\
			$\pi_i$ & \rcol{Eigenvalues of $\Pi$, with $\pi_i\ge0$ and $\sum_i\pi_i=1$} \\
			$W^{s_1s_2}_{ph/pp,O}(\mathbf k)$ & \rcol{Nonnegative particle-hole/particle-particle sum-rule weights (Eqs.~\eqref{eq:Wph-def}--\eqref{eq:Wpp-def})} \\
			$\chi_{N}$ & \rcol{Reference Pauli susceptibility $2\mu_{B}^{2}N(0)$} \\
			$K_N^{\rm spin}$ & \rcol{$A\chi_N$ for scalar hyperfine coupling (Eq.~\eqref{eq:kappa-def})} \\
			$F_s(\lambda)$ & \rcol{Pure-$s$-wave SOC interpolation kernel (Eq.~\eqref{eq:Fs})} \\
			$\lambda_{\rm eff}$ & \rcol{SOC-to-gap ratio from the averaged kernel (Eq.~\eqref{eq:lambda-eff})} \\
			$Y_\lambda(\mathbf k,T)$ & \rcol{Yosida function on helicity sheet $\lambda=\pm$} \\
			$\mathrm{Tr}_s,\mathrm{Tr}_{\rm BdG}$ & \rcol{Single-particle and BdG-space traces} \\
			$\langle\cdot\rangle_{\rm FS}$ & \rcol{Reference-Fermi-surface average normalized by $N(0)$} \\
			$\langle\cdot\rangle_\lambda$ & \rcol{Average over helicity sheet $\lambda=\pm$} \\
			\hline\hline
			\end{tabular*}
		\end{minipage}
		\end{center}
	
	\section{Finite-Field Susceptibility and Knight-Shift Geometry}
	\label{app:finite-field-response}

	In a nonzero magnetic field, the static differential susceptibility and the equilibrium polarization measured by the Knight shift are distinct response quantities. This appendix records the normal-state texture scales, defines both responses, states when the linear-response ellipsoid remains applicable, and specifies the observable to be used in a self-consistent finite-field comparison.

	\subsection{Normal-state texture and model scales}

		For the finite-field normal-state texture tensor defined in Eq.~\eqref{eq:Pi-H}, the relations $\xi_{-\mathbf k}=\xi_{\mathbf k}$ and $\mathbf g_{-\mathbf k}=-\mathbf g_{\mathbf k}$, together with the change of variable $\mathbf k\mapsto-\mathbf k$, give
	\[
		\Pi^{(b)}_{\mu\nu}(\mathbf H)
		=\Pi^{(b)}_{\mu\nu}(-\mathbf H).
	\]
	At zero field, $\Pi^{(b)}_{\mu\nu}(0)=\Pi^{(g)}_{\mu\nu}\equiv\langle\hat g_\mu\hat g_\nu\rangle_{\rm FS}$. Hence, for $\mu_BH\ll g_{\min}$, its smooth weak-field expansion begins at second order:
	\begin{equation}
		\Pi^{(b)}_{\mu\nu}(\mathbf H)
		=\Pi^{(g)}_{\mu\nu}
		+O\!\left(\mu_B^2H^2/g_{\min}^2\right).
		\label{eq:Pi-H-weak}
	\end{equation}
	This is only a normal-state geometric statement; it does not imply the same expansion for the superconducting susceptibility $\chi_{\mu\nu}(0,\mathbf H)$.

	A useful geometric landmark is the normal-state cancellation field. For the spherical model $\mathbf g_{\mathbf k}=g\mathbf k$ with $\mathbf H\parallel\hat z$,
	\begin{equation}
		\mu_BH^\ast=gk_F,
		\label{eq:Hstar}
	\end{equation}
	at which $\mathbf b_{\mathbf k}$ vanishes at one FS pole. The regime
	\begin{equation}
		|gk_F-\mu_BH|\lesssim\Delta
		\label{eq:strong-locking-criterion}
	\end{equation}
		therefore signals the breakdown of a helicity-based strong-locking expansion near that pole, but does not by itself imply a superconducting gap closing.

	Bogoliubov-Fermi-surface formation is a separate, model-dependent condition. In the spherical $s$-wave model, the minimum polar quasiparticle energy is
	\[
		E_{\min}=\Delta(T=0,g,H)-\mu_BH,
	\]
	so the onset field $H_{c2}^{Z}$ satisfies
	\begin{equation}
		\mu_BH_{c2}^{Z}=\Delta(T=0,g,H_{c2}^{Z}),
		\label{eq:BFS-onset}
	\end{equation}
	provided superconductivity remains stable. The response kink reported in Refs.~\cite{pang2025,pang2026} occurs at this onset, which need not coincide with $H^\ast$. Neither condition defines a universal finite-field locking identity or a general Zeeman-dominated oblate limit.

	\subsection{Differential and secant responses}

	Let $\bm{\eta}$ collect the superconducting order-parameter amplitudes, phases, and orientations retained in a mean-field description. The equilibrium grand potential and magnetization are
	\begin{equation}
		\Omega_\ast(T,\mathbf H)
		\equiv\Omega[T,\mathbf H,\bm{\eta}_\ast(T,\mathbf H)],
		\qquad
		M_\mu(T,\mathbf H)
		=-\frac{d\Omega_\ast}{dH_\mu},
		\label{eq:finiteH-M}
	\end{equation}
	where $\partial\Omega/\partial\eta_a=0$ at $\bm{\eta}=\bm{\eta}_\ast$. The thermodynamic differential susceptibility about the finite-field equilibrium state is
	\begin{equation}
		\chi^{\rm diff}_{\mu\nu}(T;\mathbf H)
		\equiv
		\frac{\partial M_\mu(T,\mathbf H)}{\partial H_\nu}
		=-\frac{d^2\Omega_\ast}{dH_\mu dH_\nu}.
		\label{eq:finiteH-chidiff}
	\end{equation}
		It is symmetric wherever the equilibrium state is smooth and single-valued.

	Assume first that $\mathbf M(T,\mathbf0)=0$. Integrating Eq.~\eqref{eq:finiteH-chidiff} along the straight path from $\mathbf0$ to $\mathbf H$ gives
	\begin{equation}
		\begin{aligned}
			M_\mu(T,\mathbf H)
			&=\bar\chi_{\mu\nu}(T;\mathbf H)H_\nu,\\
			\bar\chi_{\mu\nu}(T;\mathbf H)
			&\equiv\int_0^1d\lambda\,
			\chi^{\rm diff}_{\mu\nu}(T;\lambda\mathbf H).
		\end{aligned}
		\label{eq:finiteH-chisec}
	\end{equation}
	The path-averaged tensor $\bar\chi$ is the \emph{secant susceptibility} connecting the equilibrium magnetization to the applied field. If a selected state carries a spontaneous magnetization, the left-hand side must instead be replaced by $\mathbf M(T,\mathbf H)-\mathbf M(T,\mathbf0)$ and the zero-field hyperfine contribution treated separately.

		With conventional unit factors absorbed into the site-$\mathbf R$ hyperfine tensor $A^{(\mathbf R)}$, the longitudinal spin Knight shift for $\mathbf H=H\hat{\mathbf H}$ is
	\begin{equation}
		\begin{aligned}
			K_s^{(\mathbf R)}(\hat{\mathbf H},T,H)
			&=\frac{\hat{\mathbf H}\cdot
			A^{(\mathbf R)}\mathbf M(T,H\hat{\mathbf H})}{H}\\
			&=\hat{\mathbf H}\cdot A^{(\mathbf R)}
			\bar\chi(T;H\hat{\mathbf H})\cdot\hat{\mathbf H}.
		\end{aligned}
		\label{eq:finiteH-K}
	\end{equation}
	Thus a finite-field Knight shift measures the equilibrium secant response, while a Kubo correlator evaluated about a finite-field state gives a differential response. For a time-reversal-invariant state without a spontaneous moment,
	\begin{equation}
		\bar\chi(T;\mathbf H)
		=\chi^{\rm diff}(T;\mathbf0)+O(H^2),
		\label{eq:finiteH-linear}
	\end{equation}
	so Eqs.~\eqref{eq:Kdef}--\eqref{eq:Kscalar} follow in the linear-response limit.

	\subsection{Ellipsoid versus finite-field angular surface}

	The Knight-shift ellipsoid of Sec.~\ref{sec:ellipsoid} is defined by the symmetric linear-response tensor $\chi^{\rm diff}(T;\mathbf0)$. At a fixed background vector $\mathbf H_0$, the local differential tensor $\chi^{\rm diff}(T;\mathbf H_0)$ can likewise define a candidate-specific differential-response ellipsoid. This object describes the response to an additional infinitesimal probe about $\mathbf H_0$; it is not generally the equilibrium Knight shift at $\mathbf H_0$.

	For a rotation experiment at fixed field magnitude $H$, the directly measured object is the angular surface
	\begin{equation}
		\mathcal S_K(T,H)
		\equiv
		\left\{
		K_s(\hat{\mathbf H},T,H)\hat{\mathbf H}:
		\hat{\mathbf H}\in S^2
		\right\}.
		\label{eq:finiteH-surface}
	\end{equation}
	Each direction samples $\bar\chi(T;H\hat{\mathbf H})$ at a different background vector. The data admit a field-conditioned ellipsoid representation only if one symmetric tensor $K_H^{\rm eff}$, independent of $\hat{\mathbf H}$, satisfies
	\begin{equation}
		K_s(\hat{\mathbf H},T,H)
		=\hat{\mathbf H}\cdot K_H^{\rm eff}(T)\cdot\hat{\mathbf H}
		\label{eq:finiteH-effective-tensor}
	\end{equation}
	over the measured angular range. Otherwise the surface should not be assigned principal semi-axes or simplex coordinates. Measurements performed at different field magnitudes require either a demonstrated field-independent regime, interpolation to a common field, or a model-specific calculation at each actual field vector.

	\subsection{Self-consistent order-parameter response}

	Differentiating the stationarity condition for $\bm{\eta}_\ast$ gives the full thermodynamic susceptibility
	\begin{equation}
		\chi^{\rm diff}_{\mu\nu}
		=-\Omega_{H_\mu H_\nu}
		+\Omega_{H_\mu\eta_a}
		\left(\Omega_{\eta\eta}^{-1}\right)_{ab}
		\Omega_{\eta_bH_\nu},
		\label{eq:finiteH-relaxation}
	\end{equation}
	where all derivatives are evaluated at the self-consistent equilibrium state and the order-parameter Hessian is restricted to stable, nonredundant variables. A Kubo evaluation with $\bm{\eta}$ held fixed gives the first term. The second term accounts for relaxation or rotation of the gap and $\mathbf d$ vector under the infinitesimal probe; it vanishes only when the mixed derivatives vanish or the relevant order-parameter variables are externally pinned. In an exact treatment the corresponding collective vertex contributions are part of the thermodynamic response.

	Reference~\cite{pang2026} demonstrates how a finite Zeeman field can be included in the background BdG Hamiltonian and how the gap can be solved self-consistently before evaluating the Kubo response. For direct comparison with a finite-field Knight-shift experiment when the response is nonlinear, one should additionally calculate $\mathbf M(T,\mathbf H)$ from Eq.~\eqref{eq:finiteH-M}, or integrate the full differential susceptibility as in Eq.~\eqref{eq:finiteH-chisec}, and then use Eq.~\eqref{eq:finiteH-K}. Orbital and diamagnetic shifts must be removed separately. Protocol~C implements this distinction.

	\subsection{Conditional UBe$_{13}$ ellipsoid construction}
	\label{app:ube13-construction}

		The single-crystal $^{9}$Be NMR measurements of Ref.~\cite{Minami2026UBe13} illustrate both the reconstruction procedure and its finite-field limitation. In the zero-field linear-response limit, or under the additional approximation that all experimental fields probe one field-independent tensor, scalar hyperfine coupling gives
	\begin{equation}
		q(\hat{\mathbf H})
		\equiv
		1-\frac{K_s(\hat{\mathbf H},0)}
		{K_s(\hat{\mathbf H},T_c^+)}
		=\hat{\mathbf H}\cdot\Pi\cdot\hat{\mathbf H}+o(1).
		\label{eq:cubic-directional-map}
	\end{equation}
	For a general anisotropic hyperfine tensor $A$, the raw ratio is instead
	\begin{equation}
		q_A(\hat{\mathbf H})
		=\frac{\hat{\mathbf H}\cdot A\Pi\cdot\hat{\mathbf H}}
		{\hat{\mathbf H}\cdot A\cdot\hat{\mathbf H}}+o(1),
		\label{eq:cubic-directional-hf}
	\end{equation}
	and is not a Rayleigh quotient of $\Pi$. The diagonal coaxial case reduces axis by axis to the normalized electronic deficit.

		Table~II of Ref.~\cite{Minami2026UBe13} reports
	\begin{equation}
		\begin{aligned}
			q_{[001]}&=0.25\pm0.04,\\
			q_{[111]}&=0.33\pm0.05,\\
			q_{[110]}&=0.40\pm0.05.
		\end{aligned}
		\label{eq:ube13-deficits}
	\end{equation}
	If $\Pi$ is diagonal in the cubic frame, the $[111]$ direction samples the trace:
	\begin{equation}
		q_{[111]}=\tfrac13\operatorname{Tr}\Pi+o(1)
		=\tfrac13+o(1),
		\label{eq:cubic-111}
	\end{equation}
	and hence
	\begin{equation}
		\frac{\operatorname{Tr}\chi(0)}{\chi_N}
		=3(1-q_{[111]})+o(1)
		=2.01\pm0.15+o(1).
		\label{eq:ube13-trace}
	\end{equation}
	Without diagonality, Eq.~\eqref{eq:cubic-111-general} must be used.

	Under the single-domain axial ansatz of Ref.~\cite{Minami2026UBe13},
	\begin{equation}
		\Pi=\operatorname{diag}(\pi_\perp,\pi_\perp,\pi_\parallel),
		\qquad 2\pi_\perp+\pi_\parallel=1.
		\label{eq:ube13-axial}
	\end{equation}
	Then $q_{[001]}=\pi_\parallel+o(1)$ and $q_{[110]}=\pi_\perp+o(1)$, while $q_{[111]}=1/3+o(1)$. The trace check gives $2q_{[110]}+q_{[001]}=1.05\pm0.11$, and a weighted fit yields
	\begin{equation}
		(\pi_x,\pi_y,\pi_z)\simeq(0.38,0.38,0.24).
		\label{eq:ube13-simplex}
	\end{equation}
	The corresponding putative ellipsoid has semi-axes $(0.62,0.62,0.76)$ and is weakly prolate along $[001]$.

		This procedure does not reconstruct a zero-field tensor. The measurements use $2$~T for $[001]$ and $[111]$ but $1.75$~T for $[110]$, and Ref.~\cite{Minami2026UBe13} allows field-induced order-parameter mixing and domain selection. Without the fixed-tensor approximation, the three values sample the finite-field angular surface at different background vectors. A common-field angular scan, an $H\to0$ extrapolation, or a self-consistent calculation of $\mathbf M(T,\mathbf H)/H$ at each actual field vector is required to establish one ellipsoid.

	\section{Dynamical Response and the Momentum Limitation of $1/T_1$}
	\label{app:dynamical-response}
	\label{sec:spinFGT}
	
	The main text uses $1/T_1$ only as a momentum-summed transverse probe. This appendix gives the Kramers--Kronig relation, explains why the Bogoliubov sum rule is not a dynamical spectral-weight conservation law, and derives the restricted exact-$\mathbf q=0$ selection rule.
	
	\subsection{Kramers--Kronig subtraction}
	
	Let $\chi_{\mu\mu}^{R}(\mathbf q,\omega;T)$ denote the retarded susceptibility. The thermodynamic uniform susceptibility relevant to the Knight shift is
	\begin{equation}
		\chi_{\mu\mu}(0,T)
		=\frac{2}{\pi}\lim_{\mathbf q\to0}\int_0^\infty
		\frac{\operatorname{Im}\chi_{\mu\mu}^{R}(\mathbf q,\omega;T)}
		{\omega}\,d\omega .
		\label{eq:KK}
	\end{equation}
	Consequently,
	\begin{widetext}
		\begin{equation}
			\frac{2}{\pi}\lim_{\mathbf q\to0}\int_0^\infty
			\frac{d\omega}{\omega}
			\big[\operatorname{Im}\chi_{\mu\mu}^{R,N}(\mathbf q,\omega)
			-\operatorname{Im}\chi_{\mu\mu}^{R,SC}(\mathbf q,\omega)\big]
			=\chi_{\mu\mu}^{N}(0)-\chi_{\mu\mu}^{SC}(0).
			\label{eq:spinFGT}
		\end{equation}
	\end{widetext}
	The $\mathbf q\to0$ limit is essential: at exact $\mathbf q=0$, a collisionless intraband Pauli contribution can be absent from the finite-frequency spectrum. Under the zero-field strong-locking assumptions, the right-hand side is $\chi_N\Pi_{\mu\mu}+o(1)$. This is a signed subtraction, not a conservation law.
	
	\subsection{Relation to the Bogoliubov sum rule}
	
	At $T=0$, assuming no residual zero-energy quasiparticle contribution, the regular positive-frequency exact-$\mathbf q=0$ superconducting spectrum is
	\begin{equation}
		\begin{split}
			\operatorname{Im}\chi_{\mu\mu}^{R,SC}(\mathbf0,\omega>0)
			={}& \frac{\pi\mu_B^2}{2}
			\sum_{\mathbf k,s_1s_2}W_{pp,\mu}^{s_1s_2}(\mathbf k)\\
			&\times\delta(\omega-E_{\mathbf ks_1}-E_{-\mathbf ks_2}).
		\end{split}
		\label{eq:ImchiSC}
	\end{equation}
	The Bogoliubov sum rule fixes only
	\begin{equation}
		\sum_{s_1s_2}
		[W_{ph,\mu}^{s_1s_2}(\mathbf k)+W_{pp,\mu}^{s_1s_2}(\mathbf k)]
		=2,
		\label{eq:total-weight}
	\end{equation}
	not the transition energies or occupations. It therefore implies no conserved dynamical spectral weight or model-independent redistribution across $T_c$. Equation~\eqref{eq:spinFGT} fixes only the inverse-frequency moment of the $\mathbf q\to0$ spectral difference, not its sign or frequency distribution.
	\label{sec:spinFGT-rigorous}
	
	\subsection{Transverse response and momentum summation}
	\label{sec:T1-vanishing}
	
	For $\mathbf H\parallel\hat e_\alpha$, nuclear relaxation probes the transverse channels $\beta,\gamma$:
	\begin{equation}
		\begin{split}
			\left.\frac{1}{T_1T}\right|_{\mathbf H\parallel\hat e_\alpha}
			={}&\gamma_n^2k_B\sum_{\mathbf q}|A_{\rm hf}(\mathbf q)|^2\\
			&\times\lim_{\omega\to0}
			\frac{\operatorname{Im}\chi_{\beta\beta}(\mathbf q,\omega)
			+\operatorname{Im}\chi_{\gamma\gamma}(\mathbf q,\omega)}
			{\omega}.
		\end{split}
		\label{eq:T1}
	\end{equation}
	Thus ordinary $1/T_1$ cannot experimentally separate the exact $\mathbf q=0$ term from nonzero-$\mathbf q$ contributions. Even if the exact term is known theoretically to vanish, the remaining signal may be dominated by arbitrarily small nonzero momentum and need not indicate a response peaked at a finite wave vector.
	
		For bookkeeping, one may write
	\begin{equation}
		\frac{1}{T_1T}
		=\left.\frac{1}{T_1T}\right|_{\mathbf q=0}
		+\left.\frac{1}{T_1T}\right|_{\mathbf q\ne0}.
		\label{eq:T1-decomp}
	\end{equation}
	In the strong-SOC helicity regime, if the weak measurement field preserves the zero-field helicity selection rules and
	\begin{equation}
		\Pi^{(g)}_{\beta\beta}=\Pi^{(g)}_{\gamma\gamma}=0,
		\label{eq:T1-transverse-condition}
	\end{equation}
	then both transverse spin operators have no intraband helicity matrix elements. For nonzero $g_{\min}$ and a fully gapped superconducting state,
	\begin{equation}
		\boxed{\;
		\left.\left(\frac{1}{T_1T}\right)_{\mathbf q=0}\right|_{\mathbf H\parallel\hat e_\alpha}
		=0
		\;}
		\label{eq:T1-vanishing}
	\end{equation}
	in both the normal and superconducting states within these assumptions.
	
	The criterion requires both transverse projections to vanish, applies only after the tensor has been independently identified as $\Pi^{(g)}$, and additionally requires scalar hyperfine coupling or a calibrated coaxial tensor preserving the transverse channels. It does not follow from a zero-SOC triplet tensor $\Pi^{(d)}$ occupying the same simplex point. Outside these conditions, the static tensor gives no model-independent estimate of $1/T_1$.
	\label{sec:T1-noestimate}
	
	\subsection{Hebel--Slichter qualification}
	\label{sec:HS}
	
	The Hebel--Slichter peak is controlled by the quasiparticle density of states, coherence factors, impurity broadening, and the momentum structure of the hyperfine-weighted susceptibility~\cite{Leggett1975,AndersonandMorel1961}. Its presence can support conventional fully gapped coherence; its absence is not unique to unconventional pairing and establishes neither spectral transfer nor finite-$\mathbf q$ dominance.

	\section{Intermediate-SOC Derivation and Supporting Checks}
	\label{app:intermediate-support}

	\subsection{Reduction to a single integral}

	Specializing Eq.~(23c) of Ref.~\cite{pang2026} to $\mathbf H=0$ and $T=0$, writing the spin factor in tensor form, and taking the reference-Fermi-surface limit gives
	\begin{align}
		\chi_{\mu\nu}(0) &= \tfrac{\chi_{N}}{2}\big\langle[\delta_{\mu\nu}-\hat n_{\mu}\hat n_{\nu}]\tilde J(\lambda_{\mathbf{k}})\big\rangle_{\rm FS}, \label{eq:chi-FSaverage}\\
		\tilde J(\lambda) &= \int_{-\infty}^{\infty}\!d\xi\,\mathcal{F}(\xi,\lambda\Delta,\Delta), \label{eq:Jdef}
	\end{align}
	where $\lambda_{\mathbf k}=|\mathbf g_{\mathbf k}|/\Delta$. Choosing the global gauge $\Delta>0$, define $\alpha=\lambda\Delta$, $\xi_{\pm}=\xi\pm\alpha$, and $E_{\pm}=\sqrt{\xi_{\pm}^{2}+\Delta^{2}}$. The scalar kernel is
	\[
		\mathcal F(\xi,\alpha,\Delta)
		=\frac{E_{+}E_{-}-\xi_{+}\xi_{-}-\Delta^{2}}
		{E_{+}E_{-}(E_{+}+E_{-})}.
	\]

	\subsection{Closed-form evaluation}
	\label{sec:intermediate-closed}

	For $\lambda>0$, set
	\[
		\begin{aligned}
			u&=\xi/\Delta,&
			f&=\sqrt{(u+\lambda)^2+1},\\
			g&=\sqrt{(u-\lambda)^2+1},&
			a&=(f+g)/2,\\
			t&=u/a=(f-g)/(2\lambda).&&
		\end{aligned}
	\]
	Then $t\in(-1,1)$, $a^2=\lambda^2+(1-t^2)^{-1}$, and Eq.~\eqref{eq:Jdef} reduces to
	\[
		\tilde J(\lambda)
		=\int_{-1}^{1}\!dt\,
		\frac{\lambda^2(1-t^2)}
		{1+\lambda^2(1-t^2)}
		=2\left[1-\frac{\operatorname{asinh}\lambda}
		{\lambda\sqrt{1+\lambda^2}}\right].
	\]
	The value at $\lambda=0$ follows by continuity. Thus $F_s(\lambda)=\tilde J(\lambda)/2$ gives Eq.~\eqref{eq:Fs}, and substitution into Eq.~\eqref{eq:chi-FSaverage} gives Eq.~\eqref{eq:chi-explicit}.

	\subsection{Kernel properties and inversion bounds}

	The asymptotic expansions are
	\[
		\begin{aligned}
			F_s(\lambda)
			&=\tfrac23\lambda^2-\tfrac{8}{15}\lambda^4
			+O(\lambda^6), && \lambda\to0,\\
			F_s(\lambda)
			&=1-\frac{\ln(2\lambda)}{\lambda^2}
			+\frac{2\ln(2\lambda)-1}{4\lambda^4}\\
			&\quad+O\!\left(\frac{\ln\lambda}{\lambda^6}\right),
			&& \lambda\to\infty.
		\end{aligned}
	\]
	Representative values are $F_s(1)=0.376775$, $F_s(2)=0.677193$, and $F_s(5)=0.909299$. Half-saturation, $F_s(\lambda_{1/2})=1/2$, occurs at $\lambda_{1/2}=1.316811$.

	To quantify the nonuniform-SOC inversion in Eq.~\eqref{eq:lambda-eff}, define
	\[
		\bar\lambda=\langle\lambda_{\mathbf k}\rangle_{\rm FS},\qquad
		\operatorname{Var}(\lambda_{\mathbf k})
		=\langle(\lambda_{\mathbf k}-\bar\lambda)^2\rangle_{\rm FS}.
	\]
	For a narrow distribution about $\bar\lambda>0$,
	\begin{equation}
		\langle F_s(\lambda_{\mathbf k})\rangle_{\rm FS}
		=F_s(\bar\lambda)+\tfrac12F_s''(\bar\lambda)
		\operatorname{Var}(\lambda_{\mathbf k})+\cdots,
		\label{eq:Fs-variance}
	\end{equation}
	and therefore
	\[
		\lambda_{\rm eff}-\bar\lambda
		=\frac{F_s''(\bar\lambda)}{2F_s'(\bar\lambda)}
		\operatorname{Var}(\lambda_{\mathbf k})+\cdots.
	\]
		Direct evaluation gives one inflection point, $\lambda^\ast=0.668907\ldots$, with $F_s$ convex below and concave above it. Jensen's inequality yields $\lambda_{\rm eff}\ge\bar\lambda$ if the full distribution lies in $0\le\lambda_{\mathbf k}\le\lambda^\ast$, and $\lambda_{\rm eff}\le\bar\lambda$ if it lies in $\lambda_{\mathbf k}\ge\lambda^\ast$. No universal ordering follows when the distribution crosses the inflection point. Quantitative extraction of the mean SOC-to-gap ratio therefore requires independent information about the FS variation of $|\mathbf g_{\mathbf k}|$, for example from DFT+SOC, or dominance by a near-isotropic FS sheet.
		
	\subsection{Numerical cross-check against published $T=0$ endpoints}

	\begin{center}
	\begin{minipage}{\columnwidth}
		\centering
		\includegraphics[width=\columnwidth]{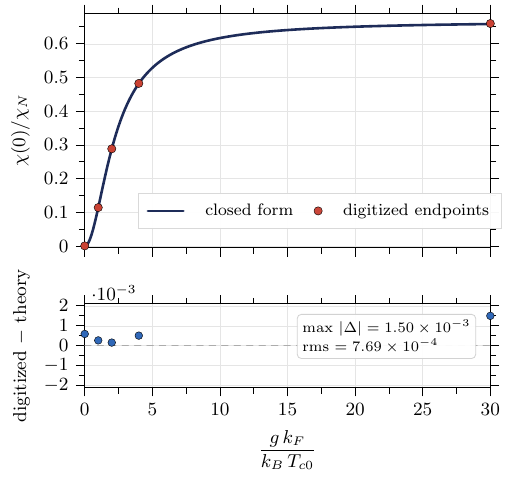}
		\captionof{figure}{Numerical cross-check of Eq.~\eqref{eq:Fs} against the digitized $T=0$ endpoints of the $s$-wave curves in Fig.~3 of Ref.~\cite{pang2026}. The solid line shows Eq.~\eqref{eq:fig3-crosscheck}, and the symbols are the digitized endpoints. The lower panel shows the digitized values minus theory.}
		\label{fig:fig3-t0-validation}
	\end{minipage}
	\end{center}
	
	The five $s$-wave curves in Fig.~3 of Ref.~\cite{pang2026} are labeled by
	$x\equiv gk_{F}/(k_{B}T_{c0})\in\{0,1,2,4,30\}$. Digitizing their $T=0$ intercepts and using the weak-coupling BCS ratio
	$\Delta_{0}/(k_{B}T_{c0})=\pi/e^{\gamma_{E}}=1.7638769\ldots$ gives
	$\lambda=xe^{\gamma_E}/\pi$ and hence
	\begin{equation}
		\frac{\chi(0)}{\chi_{N}}
		=\frac{2}{3}F_{s}\!\left(\frac{xe^{\gamma_E}}{\pi}\right).
		\label{eq:fig3-crosscheck}
	\end{equation}
	Figure~\ref{fig:fig3-t0-validation} compares this prediction with the digitized endpoints. The maximum absolute deviation is $1.5\times10^{-3}$ in $\chi/\chi_N$, and the root-mean-square deviation is $7.7\times10^{-4}$.

	\subsection{Heuristic extension to a unitary triplet with an anisotropic gap}
		
	At the level of the scalar $\xi$ integration, the replacement
	$\Delta\to|\mathbf d(\mathbf k)|$ suggests the heuristic kernel
	\[
		F_{\rm t}(\mathbf k)
		\approx F_s\!\left(\frac{|\mathbf g_{\mathbf k}|}{|\mathbf d(\mathbf k)|}\right)
	\]
	for $|\mathbf d(\mathbf k)|>0$, with limiting value $F_{\rm t}\to1$ as $|\mathbf d|\to0$ at fixed $|\mathbf g|$. This substitution retains only the gap-magnitude dependence inherited from the pure-$s$-wave integral. It omits coherence factors associated with the angular structure of $\hat{\mathbf d}$, intra-helicity triplet channels, and possible residual zero-energy contributions near nodes. It is therefore not a controlled response identity; a rigorous triplet derivation is beyond the present scope.

	\bibliographystyle{apsrev4-2}
	\bibliography{mag-soc}
	
	\end{document}